\documentclass[onecolumn,english]{IEEEtran}
\usepackage[T1]{fontenc}
\usepackage[latin9]{inputenc}
\usepackage{float}
\usepackage{amsmath}
\usepackage{amsthm}
\usepackage{amssymb}
\usepackage{graphicx}
\usepackage{setspace}
\makeatletter
\theoremstyle{plain}
\newtheorem{thm}{\protect\theoremname}
\theoremstyle{definition}
\newtheorem{defn}[thm]{\protect\definitionname}
\theoremstyle{plain}
\newtheorem{prop}[thm]{\protect\propositionname}
\theoremstyle{plain}
\newtheorem{lem}[thm]{\protect\lemmaname}

\usepackage{hyperref}
\usepackage{enumitem}
\usepackage{breqn}
\usepackage{bbm} 
\usepackage{cite}
\usepackage{tikz} 

\DeclareMathOperator*{\argmax}{arg\,max}

\DeclareMathOperator{\supp}{supp}

\global\long\def\P{\mathbb{P}}
\global\long\def\E{\mathbb{E}}

\global\long\def\I{\mathbbm{1}}

\global\long\def\trre[#1,#2]{\overset{{\scriptstyle (#2)}}{#1}} 
\renewcommand\[{\begin{equation}}
\renewcommand\]{\end{equation}}

\allowdisplaybreaks

\author{
\IEEEauthorblockN{Yoav Chachamovitz and Nir Weinberger}

\IEEEauthorblockA{The Viterbi Faculty of Electrical and Computer Engineering\\
  	    Technion--Israel Institute of Technology\\
Technion City, Haifa 3200004, Israel
} \\
\IEEEauthorblockA{syoavcha@technion.ac.il, nirwein@technion.ac.il}\\
}

\makeatother

\usepackage{babel}
\providecommand{\definitionname}{Definition}
\providecommand{\lemmaname}{Lemma}
\providecommand{\propositionname}{Proposition}
\providecommand{\theoremname}{Theorem}

\begin{document}
\title{Error Exponent Bounds for Optimal Short-Read Clustering \thanks{This research was supported by the Israel Science Foundation (ISF),
grant no. 1782/22 and the United States -- Israel Binational Science
Foundation (NSF-BSF), grant no. 2024763. A short version of this paper
was submitted to the 2026 IEEE Information Theory Workshop.}}
\maketitle
\begin{abstract}
Motivated by the operation of decoders for DNA storage, we consider
the problem of unsupervised clustering of noisy short sequences, each
generated from one of multiple possible unknown source sequences after
passing through a memoryless channel. Focusing on the statistically
optimal clustering rule, we derive upper and lower bounds on the probability
of incorrect clustering as a function of the sequence length, the
number of reads, and the channel statistics. 
\end{abstract}


\section{Introduction}

We consider the problem of unsupervised clustering of noisy short
reads. In this problem, $m$ source sequences are randomly generated,
and a total of $n$ reads from these sequences is observed. Each such
read is a noisy observation of one of the $m$ source sequences. We
assume that the $m$ clean source sequences are \emph{unobserved}
by the clustering device, and the identity of the source sequence
that generated each read is randomly chosen, independently for each
read. The goal is to cluster the $n$ reads into $m$ groups, so that
each group contains exactly the reads generated from the same source
sequence. We derive upper and lower bounds on the error exponent of
the minimal clustering error probability.  

Such a problem may arise in various engineering problems, such as
sensor networks -- where measurements from multiple sources are received
without identifiers, distributed inference -- where data streams
must be separated based on statistical similarity, and communication
systems with uncoordinated transmitters. However, as we next discuss
in detail, our main motivation comes from the \emph{DNA storage} application
\cite{churchNextGenerationDigitalInformation2012,Goldman2013,Grass2015,Yazdi2015,Kiah2016,erlichDNAFountainEnables2017,Organick2018}.
A widely adopted information-theoretic model for DNA storage is the
\emph{noisy shuffling sampling channel} \cite{shomoronyInformationTheoreticFoundationsDNA2022,ShomoronyHeckel2021}
(also called a channel with \emph{sliced information} \cite{sima2021coding}).
In this channel, the message is encoded into a multiset of $m$ (short)
strings of length $\ell$ each, from the alphabet $\{A,C,G,T\}$.
Each string is then synthesized to a DNA molecule, and the $m$ molecules
are stored in a pool, \emph{without order}. Data retrieval is modeled
by randomly sampling $n$ molecules from the pool with replacement,
and sequencing each one to obtain a noisy read of the input string
it encoded. The decoder then decides on the stored message based on
the set of $n$ output reads.

Various papers have explored information-theoretic fundamental limits
for this problem, such as capacity \cite{ShomoronyHeckel2021,lenzUpperBoundCapacity2019,lenzAchievingCapacityDNA2020,LenzSiegelWachterZehYaakobi2023,shomoronyInformationTheoreticFoundationsDNA2022,weinberger2022dna,gerzon2025capacity,tamir2025achievable}
and error probability \cite{weinberger2022error,ling2025exact,ling2025error}.
These limits are commonly derived without considering the \emph{complexity}
of the decoder. By contrast, suggestions for practical decoders are
based on an initial clustering step \cite{rashtchian2017clustering,sima2023error,boruchovsky2025dna}.
In this step, the $n$ output reads are clustered according to the
identity of the sampled molecule. Assuming that this clustering can
be achieved with high probability, this step significantly simplifies
the next steps of the decoder. With this goal in mind, the clustering
step clearly should be \emph{oblivious} to the message, and thus also
to the input molecules. While, in principle, the prior knowledge of
the codebook can be exploited, it is more convenient to model the
input molecules as random (and, as is well known, good codebooks have
distributions that are close to random).

While suboptimal clustering rules are used in practice, we focus on
\emph{optimal} clustering rules. Indeed, if such rules can be shown
to require stringent conditions for low clustering error probability,
then practical rules cannot perform better. On a positive note, if
the conditions are met, one can expect that practical algorithms can
be developed that are only marginally suboptimal. We initiate the
study of the performance of optimal clustering rules. Typically, the
analysis of the \emph{noisy shuffling sampling channel} \cite{shomoronyInformationTheoreticFoundationsDNA2022}
assumes the reads are very short compared with their number, which
is expressed by the scaling $\ell=\beta\log n$, that we adopt in
this paper. However, it is also assumed that $m$ scales with $n$,
e.g., $m=\alpha n$ for some scaling factor $\alpha$, or $n=m\log m$
\cite{weinberger2022error,ling2025error}. As shown below, the analysis
of the error probability is complicated, so we focus on the case of
a constant $m$ (which is challenging on its own). We finally mention
that this work continues our line of work on optimal statistical inference
tasks for the processing of short reads, such as reference-based reordering
and sequence-alignment \cite{weinberger2024fundamental,luria2025optimal}.

We make the following contributions. We show that under the assumed
scaling the error probability decays exponentially with $\ell$ (and
since $n=e^{\ell/\beta}$ it decays polynomially with $n$), and therefore
aim to characterize the error exponent with respect to (w.r.t.) $\ell$.
For achievability, we derive an upper bound on the clustering error
probability of the optimal clustering rule, which holds for any $\ell$
large enough (Proposition \ref{prop: Raw upper bound on the error probability}).
The proof of this bound is the main technical contribution of the
paper. We then use this bound to derive a single-letter lower bound
on the error exponent (Theorem \ref{thm: Exponent for a uniform source}).
We also derive an upper bound on the error exponent (Theorem \ref{thm: upper bound on the exponent}),
via a reduction of the clustering problem to an assignment problem
and a connection to channel coding in the random-coding regime \cite[Ch. 5]{gallager1968information}.

The outline of the rest of the paper is as follows. In Section \ref{sec:Problem-Formulation}
we state notation conventions and formulate the problem, and also
discuss optimal clustering rules and various modifications. In Section
\ref{sec:Bounds-on-the} we state our main results. We first derive
a finite-length upper bound on the clustering error probability, and
then evaluate its asymptotics in order to obtain a lower bound on
the error exponent. We then state and prove an upper bound on the
error exponent. In Section \ref{sec:Proof-of- upper bound on error probability}
we prove the finite-length bound on the error probability, which is
the main technical contribution of the paper. In Section \ref{sec:conclusion}
we conclude the paper, and in the appendices, we provide the remaining
proofs. 

\section{Problem Formulation \label{sec:Problem-Formulation}}

\subsection{Notation Conventions \label{subsec:Notation-Conventions}}

Let ${\cal X}$ be a finite alphabet and let ${\cal P}({\cal X})$
be the set of all probability mass functions (PMFs) on ${\cal X}$
(i.e., the $(|{\cal X}|-1)$-dimensional probability simplex). Let
$x\in{\cal X}^{\otimes\ell}$ denote a vector (sequence) of length
$\ell$, given as $x=(x(1),\ldots,x(\ell))$. Let ${\cal P}_{\ell}({\cal X})$
denote the set of all types (empirical distributions) of length $\ell$.
Let $T_{\ell}(Q_{X})$ denote the type class \cite[Ch. 2]{csiszar2011information}
of a type $Q_{X}\in{\cal P}_{\ell}({\cal X})$, that is, the set of
all empirical PMFs for length $\ell$ vectors over ${\cal X}$. For
a pair $x,\tilde{x}\in{\cal X}^{\otimes\ell}$, let $d_{\text{Ham}}(x,\tilde{x})=\sum_{i=1}^{\ell}\I\{x(i)\neq\tilde{x}(i)\}$
denote the Hamming distance. For a given pair of distributions $P_{X},Q_{X}\in{\cal P}({\cal X})$,
let $D_{\text{KL}}(P_{X}\mid\mid Q_{X})$ denote the Kullback--Leibler
(KL) divergence. Let $h_{\text{bin}}(t):=-t\cdot\log t-(1-t)\cdot\log(1-t)$
for $t\in(0,1)$ and $h_{\text{bin}}(0)=h_{\text{bin}}(1)=0$ denote
the binary entropy function. For an integer $\ell\in\mathbb{N}$,
let $[\ell]=\{1,2,\ldots,\ell\}$. For $a,b\in\mathbb{R}$, let their
maximum (resp. minimum) be denoted as $a\vee b$ (resp. $a\wedge b$).
Let the complement of a set ${\cal A}$ be denoted as ${\cal A}^{c}$.
Let logarithms and exponents be taken with an arbitrary, yet fixed,
base. Let $\equiv$ denote equivalence, mainly used to locally simplify
notation. 

\subsection{System Model \label{subsec:System-Model}}

Let ${\cal X}$ be a finite alphabet, and assume that $m$ short source
sequences are randomly drawn $X_{1}^{m}=(X_{1},\ldots,X_{m})$ where
$X_{i}=(X_{i}(1),X_{i}(2),\ldots,X_{i}(\ell))\in{\cal X}^{\otimes\ell}$.
The source sequences are drawn independently and identically distributed
(i.i.d.) according to some $X_{i}\sim P_{X}^{(\ell)}\in{\cal P}({\cal X}^{\otimes\ell})$.
Let ${\cal Y}$ be a finite read alphabet, and assume that $Y_{1}^{n}=(Y_{1},\ldots,Y_{n})$
are $n$ independent noisy reads of random source sequences, chosen
with replacement, where $Y_{i}=(Y_{i}(1),Y_{i}(2),\ldots,Y_{i}(\ell))\in{\cal Y}^{\otimes\ell}$.
Specifically, let $S_{1}^{n}=(S_{1},\ldots,S_{n})\in[m]^{\otimes n}$
be the \emph{sampling index vector}, which is a vector of i.i.d. random
variables, where the random index is drawn as $S_{j}\sim\text{Uniform}[m]$.
Then, the $j$th read is randomly drawn as the output of a Markov
kernel $\{W^{(\ell)}(y\mid x)\}_{x\in{\cal X}^{\otimes\ell},y\in{\cal Y}^{\otimes\ell}}$,
so that $Y(j)\sim W^{\otimes\ell}(\cdot\mid X_{S_{j}})$. 

\begin{figure}[t]
\centering
\begin{tikzpicture}[x=1.2cm,y=1.2cm,>=latex]
    \tikzstyle{latent}=[draw, rectangle, rounded corners, fill=blue!15, minimum width=1.2cm, minimum height=0.8cm]
    \tikzstyle{observed}=[draw, rectangle, fill=orange!20, minimum width=1.2cm, minimum height=0.8cm]
    \tikzstyle{channel}=[->, thick]

    \node[latent] (X1) at (0,0) {$X_1$};
    \node[latent] (X2) at (3,0) {$X_2$};
    \node[latent] (X3) at (6,0) {$X_3$};

    \node[observed] (Y1) at (-0.5,-1.8) {$Y_1$};
    \node[observed] (Y3) at (0.8,-3.3) {$Y_3$};

    \node[observed] (Y2) at (2.5,-1.8) {$Y_2$};
    \node[observed] (Y5) at (3.8,-3.3) {$Y_5$};

    \node[observed] (Y4) at (5.5,-1.8) {$Y_4$};
    \node[observed] (Y6) at (6.8,-3.3) {$Y_6$};

    \draw[channel] (X1) -- (Y1) node[midway,right,xshift=2pt] {$W$};
    \draw[channel] (X1) to[out=-90,in=90] (Y3);

    \draw[channel] (X2) -- (Y2) node[midway,right,xshift=2pt] {$W$};
    \draw[channel] (X2) to[out=-90,in=90] (Y5);

    \draw[channel] (X3) -- (Y4) node[midway,right,xshift=2pt] {$W$};
    \draw[channel] (X3) to[out=-90,in=90] (Y6);
\end{tikzpicture}
\caption{Model illustration. Each source sequence $X_i\in\mathcal{X}^{\otimes \ell}$ produces reads $Y_i\in\mathcal{Y}^{\otimes \ell}$ via a noisy channel $W$ (arrows). The index mapping $S_j=i$ indicates that the $i$th source sequence generated observation $j$ (here $S_{1}^{6}=B_{1}^{6}=(1,2,1,3,2,3)$).}
\label{fig:gen_model}
\end{figure}
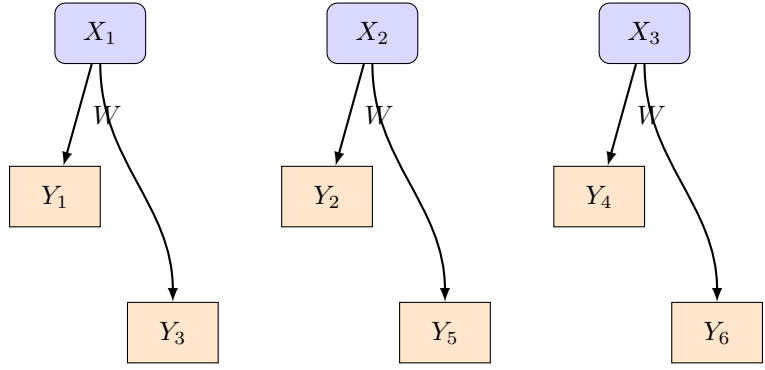

The goal of a clustering algorithm is to partition the $n$ output
reads into groups such that the reads in each group are observations
of the same source sequence, without observing the clean source sequences.
Equivalently, this amounts to the detection of the sampling index
vector $S_{1}^{n}$ up to a permutation of $[m]$. Formally, let $\Pi_{m}$
denote the symmetric group of $[m]$ (the set of all permutations
of $[m]$), and for $\pi\in\Pi_{m}$ let 
\[
\pi(S_{1}^{n}):=\left(\pi(S_{1}),\pi(S_{2}),\ldots,\pi(S_{n})\right),
\]
that is, the permutation operates on each index separately. A \emph{clusterer}
is thus $\mathsf{C}\colon({\cal Y}^{\otimes\ell})^{\otimes n}\to[m]^{\otimes n}$,
and its error probability is given by 
\begin{equation}
p_{\text{error}}(\mathsf{C})=\P\left[\bigcap_{\pi\in\Pi_{m}}\left\{ \mathsf{C}(Y_{1}^{n})\neq\pi(S_{1}^{n})\right\} \right].\label{eq: error probability definition}
\end{equation}
We consider the regime in which $\ell\equiv\ell_{n}=\beta\log n$
for some fixed scaling-length parameter $\beta$, constant $m\geq2$,
and memoryless sources and channels, that is \textbf{
\[
P_{X}^{(\ell)}=\prod_{i=1}^{\ell}P_{X}^{(1)}(x(i))
\]
}for some $P_{X}\in{\cal P}({\cal X})$, $P_{X}\equiv P_{X}^{(1)}$
(with a slight abuse of notation), and memoryless reading channel
\textbf{
\[
W^{\otimes\ell}(y\mid x)=\prod_{i=1}^{\ell}W^{(1)}(y(i)\mid x(i))
\]
}for some $\{W(y\mid x)\}_{x\in{\cal X},y\in{\cal Y}}$, $W\equiv W^{(1)}$
(with a slight abuse of notation). We assume without loss of generality
(w.l.o.g.) that $\supp(P_{X})={\cal X}$. Our goal is to characterize
the optimal scaling of the error probability\textbf{. }As we shall
see, in the considered regime the error probability decays exponentially
with $\ell$ (or polynomially with $n$) and so our goal is to derive
lower (achievability) and upper (converse) bounds on the exponent
\begin{equation}
\phi(\beta,P_{X},W):=\lim_{\ell\to\infty}-\frac{1}{\ell}\log\min_{\mathsf{C}}p_{\text{error}}(\mathsf{C}).\label{eq: error exponent definition}
\end{equation}

\subsection{Optimal Clustering Rules \label{subsec:Optimal-Clustering-Rules}}

In this section, we describe the optimal clustering rule. \textbf{ }As
noted, the clusterer observes the $n$ reads $Y_{1}^{n}$ (but not
the source sequences $X_{1}^{m}$) and its goal is to cluster the
reads according to the input that generated them, so the clustering
is invariant to permutations (see the error probability definition
\ref{eq: error probability definition}). For example, in the case
$m=2$, the sampling index vectors $S_{1}^{5}=(2,1,2,1,1)$ and $\tilde{S}_{1}^{5}=(1,2,1,2,2)$
lead to the same clustering, since the clusterer does not know the
input, and cannot distinguish $(X_{1},X_{2})$ from $(X_{2},X_{1})$.
Thus, any permutation of the labeling of the reads leads to equivalent
clustering. To avoid this ambiguity, we define the \emph{canonical
sampling index vector} as the \emph{pattern} of $S_{1}^{n}$, as follows.
\begin{defn}
Let ${\cal A}$ be a finite alphabet and let $a_{1}^{n}\in{\cal A}^{\otimes n}$.
Let the index of the first appearance of $\overline{a}\in{\cal A}$
in $a_{1}^{n}$ be 
\[
f_{a_{1}^{n}}(\overline{a}):=\min\left\{ i\in[n]\colon a_{i}=\overline{a}\right\} .
\]
Let $\psi:{\cal A}\to[|{\cal A}|]$ be a bijection that satisfies
\[
f_{a_{1}^{n}}(\overline{a})<f_{a_{1}^{n}}(\tilde{a})\Rightarrow\psi(\overline{a})<\psi(\tilde{a}).
\]
 The pattern of $a_{1}^{n}$ is defined as 
\[
\Psi(a_{1}^{n}):=(\psi(a_{1}),\psi(a_{2}),\ldots,\psi(a_{n}))\in[|{\cal A}|]^{\otimes n}.
\]
For example, for the alphabet ${\cal A}=\{\mathsf{a},\mathsf{b},\mathsf{c},\mathsf{d},\mathsf{e},\mathsf{f}\}$,
the pattern of $(\mathsf{f},\mathsf{b},\mathsf{f},\mathsf{a},\mathsf{f},\mathsf{e},\mathsf{d})$
is $(1,2,1,3,1,4,5)$. Hence, the equivalent goal of the clusterer
is to detect $B_{1}^{n}=\Psi(S_{1}^{n})$. It should be noted that
$B_{1}^{n}$ is not an i.i.d. vector (even though $S_{1}^{n}$ is),
as, for example, its first component is always $1$. Let $\Pi(n,m)$
be the set of all pattern vectors of length $n$ with components in
$[m]$. Therefore, an equivalent formulation of a clusterer is $\mathsf{C}\colon({\cal Y}^{\otimes\ell})^{\otimes n}\to\Pi(n,m)$,
and its error probability is given by 
\[
p_{\text{error}}(\mathsf{C})=\P\left[\mathsf{C}(Y_{1}^{n})\neq B_{1}^{n}\right].
\]
Since the number of possible patterns is finite, i.e., $|\Pi(n,m)|<\infty$,
the clusterer is, in principle, a multiple hypothesis tester, where
each hypothesis corresponds to a possible pattern. Thus, the optimal
clusterer is given by the \emph{maximum a posteriori} (MAP) rule
\[
\mathsf{C}_{\text{MAP}}(y_{1}^{n})=\argmax_{b_{1}^{n}\in\Pi(n,m)}\P\left[B_{1}^{n}=b_{1}^{n}\mid Y_{1}^{n}=y_{1}^{n}\right].
\]
For a given $b_{1}^{n}\in\Pi(n,m)$, the likelihood is obtained by
marginalizing over the unknown input symbols. Since the input symbols
are i.i.d., the likelihood is identical for each $s_{1}^{n}$ for
which $\Psi(s_{1}^{n})=b_{1}^{n}$. Thus, for such $s_{1}^{n}\in[m]^{\otimes n}$,
\begin{align}
\lambda(y_{1}^{n}\mid b_{1}^{n}) & :=\P\left[Y_{1}^{n}=y_{1}^{n}\mid B_{1}^{n}=b_{1}^{n}\right]\\
 & =\P\left[Y_{1}^{n}=y_{1}^{n}\mid S_{1}^{n}=s_{1}^{n}\right]\\
 & =\sum_{x_{1}^{m}\in({\cal X}^{\otimes\ell})^{\otimes m}}\P\left[X_{1}^{m}=x_{1}^{m}\mid S_{1}^{n}=s_{1}^{n}\right]\P\left[Y_{1}^{n}=y_{1}^{n}\mid S_{1}^{n}=s_{1}^{n},X_{1}^{m}=x_{1}^{m}\right]\\
 & \trre[=,a]\sum_{x_{1}^{m}\in({\cal X}^{\otimes\ell})^{\otimes m}}\P\left[X_{1}^{m}=x_{1}^{m}\right]\P\left[Y_{1}^{n}=y_{1}^{n}\mid S_{1}^{n}=s_{1}^{n},X_{1}^{m}=x_{1}^{m}\right]\\
 & \trre[=,b]\sum_{x_{1}^{m}\in({\cal X}^{\otimes\ell})^{\otimes m}}\prod_{i=1}^{m}P_{X}^{\otimes\ell}(x_{i})\prod_{j\in[n]\colon S_{j}=i}W^{(\ell)}(y_{j}\mid x_{i})\\
 & =\prod_{i=1}^{m}\sum_{x\in{\cal X}^{\otimes\ell}}P_{X}^{\otimes\ell}(x)\prod_{j\in[n]\colon S_{j}=i}W^{(\ell)}(y_{j}\mid x),\label{eq: likelihood}
\end{align}
where $(a)$ follows since $X_{1}^{m}$ is independent of $S_{1}^{n}$.
\end{defn}
As for the prior, note that if $b_{1}^{n}=\Psi(s_{1}^{n})$ for a
given sampling index vector $s_{1}^{n}$, then $\kappa(b_{1}^{n}):=\max\{b_{1},\ldots,b_{n}\}$
is the number of distinct source indices that were sampled from $[m]$.
Thus, for $b_{1}^{n}\in\Pi(n,m)$ it holds that 
\begin{align}
\P[B_{1}^{n}=b_{1}^{n}] & =\sum_{s_{1}^{n}\in[m]^{\otimes n}\colon\Psi(s_{1}^{n})=b_{1}^{n}}\P[S_{1}^{n}=s_{1}^{n}]\\
 & \trre[=,a]\frac{1}{m^{n}}\cdot\left|\left\{ s_{1}^{n}\in[m]^{\otimes n}\colon\Psi(s_{1}^{n})=b_{1}^{n}\right\} \right|\\
 & \trre[=,b]\frac{1}{m^{n}}{n \brace \kappa(b_{1}^{n})}
\end{align}
where $(a)$ holds since $S_{1}^{n}$ is distributed uniformly over
$[m]^{\otimes n}$, i.e., $\P[S_{1}^{n}=s_{1}^{n}]=\frac{1}{m^{n}}$
for each $s_{1}^{n}\in[m]^{\otimes n}$ , and $(b)$ holds since the
number of ways to partition a set of $n$ objects into $k$ non-empty
subsets is given by the Stirling number of the second kind ${n \brace k}$. 

For $k\in[m]$, let 
\[
\Pi(n,m,k):=\left\{ b_{1}^{n}\in\Pi(n,m)\colon\kappa(b_{1}^{n})=k\right\} 
\]
be the set of patterns with components $[k]$, so that $\Pi(n,m)=\bigcup_{k\in[m]}\Pi(n,m,k)$.
Hence, the MAP clustering rule is given by 
\begin{align}
\mathsf{C}_{\text{MAP}}(y_{1}^{n}) & =\argmax_{k\in[m],b_{1}^{n}\in\Pi(n,m,k)}\P[B_{1}^{n}=b_{1}^{n}]\times\P\left[Y_{1}^{n}=y_{1}^{n}\mid B_{1}^{n}=b_{1}^{n}\right]\\
 & =\argmax_{k\in[m],b_{1}^{n}\in\Pi(n,m,k)}{n \brace k}\cdot\prod_{i=1}^{m}\sum_{x\in{\cal X}^{\otimes\ell}}P_{X}^{\otimes\ell}(x)\prod_{j\in[n]\colon b_{j}=i}W^{(\ell)}(y_{j}\mid x).\label{eq: MAP rule}
\end{align}
Although optimal, such a clustering rule is fairly complicated to
compute, since even just computing the score of a single potential
clustering requires marginalization over ${\cal X}^{\otimes\ell}$.
A simple way to alleviate this is to replace the summation over $x\in{\cal X}^{\otimes\ell}$
with a maximization, leading to a joint maximization over both $x_{1}^{m}\in({\cal X}^{\otimes\ell})^{\otimes m}$
and $b_{1}^{n}\in\Pi(n,m)$. Furthermore, in some cases, $\kappa(B_{1}^{n})$
concentrates rapidly around some value. For example, in the case $m=2$,
it holds that $\kappa(B_{1}^{n})=2$ unless all reads are from the
same source sequence, and this occurs with probability $2^{-n}$.
Thus, we may consider a suboptimal clusterer that  restricts $k$
to some subset ${\cal M}_{0}$ of $[m]$. In this case, the prior
probability has a smaller effect on the error probability, and so
we may ignore it, and consider an \emph{approximate maximum likelihood}
(AML) clusterer 
\[
\mathsf{C}_{\text{AML}}(y_{1}^{n};{\cal M}_{0})=\argmax_{k\in{\cal M}_{0},b_{1}^{n}\in\Pi(n,m,k)}\prod_{i=1}^{m}\sum_{x\in{\cal X}^{\otimes\ell}}P_{X}^{\otimes\ell}(x)\prod_{j\in[n]\colon b_{j}=i}W^{(\ell)}(y_{j}\mid x),
\]
which is only approximate since ${\cal M}_{0}\neq[m]$. Moreover,
we may further restrict the maximization to a subset ${\cal B}_{n,m}\subset\cup_{k\in{\cal M}_{0}}\Pi(n,m,k)$
of highly likely patterns, for example, those for which the number
of samples from each source sequence is roughly the same, e.g., around
the average $n/m$. With a slight abuse of notation, we may consider
the clustering rule
\begin{equation}
\mathsf{C}_{\text{AML}}(y_{1}^{n};{\cal B}_{n,m})=\argmax_{b_{1}^{n}\in{\cal B}_{n,m}}\prod_{i=1}^{m}\sum_{x\in{\cal X}^{\otimes\ell}}P_{X}^{\otimes\ell}(x)\prod_{j\in[n]\colon b_{j}=i}W^{(\ell)}(y_{j}\mid x).\label{eq: approximate ML rule second}
\end{equation}

\section{Bounds on the Clustering Error Probability \label{sec:Bounds-on-the}}

Let the \emph{Bhattacharyya coefficient} for $a,\tilde{a}\in{\cal X}$
be given by 
\begin{equation}
B(a,\tilde{a})\equiv B_{W}(a,\tilde{a}):=\sum_{y\in{\cal Y}}\sqrt{W(y\mid a)W(y\mid\tilde{a})}\label{eq: Bhattacharyya coefficient}
\end{equation}
and let the \emph{Bhattacharyya distance} $d_{B}(a,\tilde{a}):=-\log B(a,\tilde{a})$.
For $x_{1},x_{2}\in{\cal X}^{\otimes\ell}$ extend it additively as
\[
D_{B}(x_{1},x_{2}):=\sum_{i=1}^{\ell}d_{B}(x_{1}(i),x_{2}(i)).
\]
With a slight abuse of notation, if $Q_{X_{1}X_{2}}\in{\cal P}_{\ell}({\cal X}^{\otimes2})$
is a joint type, then we denote 
\begin{equation}
D_{B}(Q_{X_{1}X_{2}}):=\E_{X_{1},X_{2}\sim Q}d_{B}(X_{1},X_{2})\label{eq: average Bhattacharyya coefficient}
\end{equation}
 when $(x_{1},x_{2})\in T_{\ell}(Q_{X_{1}X_{2}})$. For the given
reading channel $W$, let the maximal and minimal \emph{Bhattacharyya
distance} be defined as
\[
d_{B,\text{max}}:=\max_{a,\tilde{a}\in{\cal X}}d_{B}(a,\tilde{a}),
\]
and 
\[
d_{B,\text{min}}:=\min_{a,\tilde{a}\in{\cal X}\colon a\neq\tilde{a}}d_{B}(a,\tilde{a}),
\]
where we assume w.l.o.g. that $d_{B,\text{min}}>0$. We further denote
\[
p_{\text{min}}:=\min_{a\in{\cal X}}P_{X}(a),
\]
where we assume w.l.o.g. that $p_{\text{min}}>0$.

\subsection{Fixed Length Upper Bound \label{subsec:Fixed-Length-Upper}}

We begin with the following result, which constitutes the main technical
effort for the proof of the upper bound on the error probability. 
\begin{prop}
\label{prop: Raw upper bound on the error probability}Assume that
$d_{B,\text{\emph{max}}}<\infty$. Define for $x_{1}^{m}\in({\cal X}^{\otimes\ell})^{\otimes m}$
\[
\overline{B}(x_{1}^{m}):=\max_{i,j\in[m]\colon i\neq j}B(x_{i},x_{j}),
\]
and let ${\cal F}$ be any set such that ${\cal F}\subseteq({\cal X}^{\otimes\ell})^{\otimes m}$.
There exist $\delta_{0},\eta_{0}\in(0,1/4)$ such that for any $\delta\in(0,\delta_{0})$
and $\eta\in(0,\eta_{0})$, and for all $n\geq n_{0}(m,\beta,\delta,\eta,|{\cal X}|,p_{\text{\emph{min}}},d_{B,\text{\emph{max}}},d_{B,\text{\emph{min}}})$
\begin{align}
p_{\text{\emph{error}}}(\mathsf{C}_{\text{\emph{MAP}}}) & \leq e^{c_{0}(\delta)m\ell}\E_{X_{1}^{m}\sim P_{X}^{\otimes\ell m}}\left[\sqrt{\frac{\left(P_{X}^{\otimes\ell}(X_{1})\right)^{m-1}}{\prod_{i=2}^{m}P_{X}^{\otimes\ell}(X_{i})}}\left[\left(1+e^{c_{1}(\delta)\ell}\overline{B}(X_{1}^{m})\right)^{n}-1\right]\cdot\I\{X_{1}^{m}\in{\cal F}\}\right]\nonumber \\
 & \hphantom{===}+\P[X_{1}^{m}\in{\cal F}^{c}]+e^{-c_{2}(\delta)n\ell}+e^{c_{1}(\delta)n\ell}\left(\overline{B}(x_{1}^{m})\right)^{\frac{n}{2m}}+me^{-c_{3}\eta^{2}n},\label{eq: upper bound on clustering error finite ell}
\end{align}
and 
\[
c_{0}(\delta):=\delta\log(1/p_{\text{\emph{min}}})+h_{\text{\emph{bin}}}(\delta),
\]
\[
c_{1}(\delta)=2\delta d_{B,\text{\emph{max}}},
\]
\[
c_{2}(\delta):=\frac{\delta d_{B,\text{\emph{min}}}}{4m},
\]
and $c_{3}>0$ is a numerical constant.
\end{prop}
The proof of Proposition \ref{prop: Raw upper bound on the error probability}
appears in Section \ref{sec:Proof-of- upper bound on error probability}. 

\paragraph*{Discussion}

In the regime we consider, where $\ell=\beta\log n$, the bound in
(\ref{eq: upper bound on clustering error finite ell}) is a multi-letter
bound, as it requires computing the expectation over $P_{X}^{\otimes\ell m}$.
In what follows, we use this upper bound to compute the exponent $\phi(\beta,P_{X},W)$
(as defined in (\ref{eq: error exponent definition})). As noted,
we expect that the error probability will decay exponentially with
$\ell$. This is intuitive because in the case of $m=2$ and a uniform
source $P_{X}(x)=\frac{1}{|{\cal X}|}$, there is a probability of
$e^{-(\log|{\cal X}|)\cdot\ell}$ that $X_{1}=X_{2}$, and similarly,
that $X_{1}$ and $X_{2}$ are very close (say, different by a single
letter). In this case, clustering error is expected. Our bound (\ref{eq: upper bound on clustering error finite ell})
consists of four terms, whose origins are explained in the proof outline.
The last two terms of the bound, that is, $e^{-c_{1}(\delta)n\ell}+me^{-c_{3}\eta^{2}n}$
decay super-exponentially with $\ell$ and do not affect the exponent.
The set ${\cal F}$ will be chosen so that the exponent of $\P[X_{1}^{m}\in{\cal F}^{c}]$
is finite, and so the overall exponent will balance between this exponent
and the exponent of the term in the expectation. 

\paragraph*{Proof outline}

We begin by considering the randomness of the sampling index vector
$S_{1}^{n}$. Since the number of source sequences is $m$, and the
number of reads is $n$, we expect that each of the $m$ source sequences
will be sampled roughly $n/m$ times. For a fixed $m$ and increasing
$n$, a simple application of Hoeffding's inequality and the union
bound shows that the histogram of the number of samples from each
source sequence concentrates rapidly around this expected value $[(1-\eta)\frac{n}{m},(1+\eta)\frac{n}{m}]$,
uniformly over $i\in[m]$. We denote this ``good'' event by ${\cal G}_{\eta}$,
and show that the probability it does not occur decays as $me^{-c_{3}\eta^{2}n}$,
which corresponds to the last term of the bound (\ref{eq: upper bound on clustering error finite ell}). 

Given this observation, we define an AML  clustering rule $\mathsf{C}_{\text{AML}}^{\sharp}$,
in the spirit of (\ref{eq: approximate ML rule second}), which only
outputs balanced clusters, that is, the size of each cluster is $[(1-\eta)\frac{n}{m},(1+\eta)\frac{n}{m}]$.
Under this clustering rule, there is no need to account for possible
errors events due to imbalanced clustering of the reads. Naturally,
as this is a suboptimal rule, its error probability upper bounds the
error probability of the MAP rule $\mathsf{C}_{\text{MAP}}$ (\ref{eq: MAP rule}).
Next, given a chosen set ${\cal F}$, we upper bound the error probability
as
\[
p_{\text{error}}(\mathsf{C}_{\text{AML}}^{\sharp})\leq\max_{s_{1}^{n}}\P\left[{\cal E}\cap{\cal F}\mid S_{1}^{n}=s_{1}^{n}\right]+me^{-c\eta^{2}n}+\P\left[X_{1}^{m}\not\in{\cal F}\right],
\]
where the maximization is over the possible $s_{1}^{n}$ whose histogram
belongs to the good set ${\cal G}_{\eta}$. Next, using a standard
union bound, the error probability $\P[{\cal E}\cap{\cal F}\mid S_{1}^{n}=s_{1}^{n}]$
is upper bounded by the sum of all possible pairwise error probabilities
of the form (omitting some details for clarity of exposition)
\[
\P\left[\lambda(Y_{1}^{n}\mid\tilde{s}_{1}^{n})\geq\lambda(Y_{1}^{n}\mid s_{1}^{n})\mid S_{1}^{n}=s_{1}^{n}\right],
\]
where $\tilde{s}_{1}^{n}$ is an alternative sampling index vector
that leads to a different clustering than $s_{1}^{n}$. Due to the
permutation invariance of clustering (that is, if $s_{1}^{n}$ and
$\tilde{s}_{1}^{n}$ have the same pattern $\Psi(s_{1}^{n})=\Psi(\tilde{s}_{1}^{n})$),
this involves a delicate counting argument. In Lemma \ref{lem: Bhat upper bound on the pairwise error probability},
we then derive a Bhattacharyya-style upper bound on this pairwise
error probability. In principle, this upper bound should depend on
the Bhattacharyya coefficient between $\lambda(Y_{1}^{n}\mid\tilde{s}_{1}^{n})$
and $\lambda(Y_{1}^{n}\mid s_{1}^{n})$. However, each such likelihood
is obtained by a marginalization over $x_{1}^{m}$, see (\ref{eq: likelihood}),
and therefore the Bhattacharyya coefficient seems intractable to analyze.\footnote{See \cite{averbuch2019expurgated} for similar bounds and discussion
in the setting of coded communication over the asymmetric broadcast
channel, and \cite[Ch. 4]{merhav2025toolbox} for related techniques.} We thus further upper bound this Bhattacharyya coefficient as (again,
with some details omitted for clarity of exposition), 

\begin{align}
 & \sum_{y_{1}^{n}\in({\cal Y}^{\otimes\ell})^{\otimes n}}\sqrt{\lambda(Y_{1}^{n}\mid s_{1}^{n})\lambda(Y_{1}^{n}\mid\tilde{s}_{1}^{n})}\nonumber \\
 & =\sum_{y_{1}^{n}}\sqrt{\sum_{x_{1}^{m}\in({\cal X}^{\otimes\ell})^{\otimes m}}P_{X}^{\otimes\ell m}(x_{1}^{m})\prod_{j\in[n]}W^{(\ell)}(y_{j}\mid x_{s_{j}})\sum_{\tilde{x}_{1}^{m}\in({\cal X}^{\otimes\ell})^{\otimes m}}P_{X}^{\otimes\ell m}(\tilde{x}_{1}^{m})\prod_{j\in[n]}W^{(\ell)}(y_{j}\mid\tilde{x}_{\tilde{s}_{j}})}\\
 & \leq\sum_{x_{1}^{m}\in({\cal X}^{\otimes\ell})^{\otimes m}}\sum_{\tilde{x}_{1}^{m}\in({\cal X}^{\otimes\ell})^{\otimes m}}\sqrt{P_{X}^{\otimes\ell m}(x_{1}^{m})P_{X}^{\otimes\ell m}(\tilde{x}_{1}^{m})}\sum_{y_{1}^{n}\in({\cal Y}^{\otimes\ell})^{\otimes n}}\sqrt{\prod_{j\in[n]}W^{(\ell)}(y_{j}\mid x_{s_{j}})W^{(\ell)}(y_{j}\mid\tilde{x}_{\tilde{s}_{j}})}\\
 & =\sum_{x_{1}^{m}\in({\cal X}^{\otimes\ell})^{\otimes m}}\sum_{\tilde{x}_{1}^{m}\in({\cal X}^{\otimes\ell})^{\otimes m}}\sqrt{P_{X}^{\otimes\ell m}(x_{1}^{m})P_{X}^{\otimes\ell m}(\tilde{x}_{1}^{m})}\prod_{j\in[n]}\underbrace{\sum_{y\in{\cal Y}^{\otimes\ell}}\sqrt{W^{(\ell)}(y\mid x_{s_{j}})W^{(\ell)}(y\mid\tilde{x}_{\tilde{s}_{j}})}}_{:=B(x_{s_{j}},\tilde{x}_{\tilde{s}_{j}})}.
\end{align}
Evidently, the Bhattacharyya coefficient of interest is upper bounded
by the Bhattacharyya coefficient between the source sequences $x_{1}^{m}$,
and alternative source sequences $\tilde{x}_{1}^{m}$ (which are ``virtual''
in the sense that they are used for marginalization in the computation
of the alternative sampling index vector $\tilde{s}_{1}^{n}$). When
performing the summation over all possible $\tilde{x}_{1}^{m}$, this
upper bound can be brought to a separable form $\prod_{i=1}^{m}f_{i}(x_{1}^{m})$
where each $f_{i}(x_{1}^{m})$ depends on summation over $\tilde{x}\in{\cal X}^{\otimes\ell}$
of terms that depend only on $\{B(x_{i},\tilde{x})\}_{i\in[m]}$.
We separate this summation into two parts. In principle, the first
case is $\tilde{x}=x_{i}$ for some $x_{i}$ (a total of $m$ terms),
and so $B(x_{i},\tilde{x})=1$, and the other is the complementary
one, in which $B(x_{i},\tilde{x})<1$. For technical reasons related
to the asymptotic computation of the error exponent (see Theorem \ref{thm: Exponent for a uniform source}),
we need to ensure that in the second case $B(x_{i},\tilde{x})<e^{-\delta\ell}$
for an arbitrarily small exponent $\delta>0$. To achieve this, we
slightly modify the previous cases. Rather than considering the cases
$\tilde{x}=x_{i}$, we consider the case in which $\tilde{x}$ belongs
to a small Hamming ball around one of the $x_{i}$'s. This only slightly
changes the exponent of $B(x_{i},\tilde{x})$. The contribution of
these two types of cases can then be substituted back to the pairwise
error probability upper bound, and then back to the upper bound on
the error probability. It turns out that the contribution of the second
case is negligible, and this is the source of the term $e^{-c_{1}(\delta)n\ell}$
in (\ref{eq: upper bound on clustering error finite ell}), which,
evidently, decays super-exponentially with $\ell$. The contribution
of the $m$ terms from the first case is dominating, and depends only
on $\{B(X_{i},X_{j})\}_{i,j\in[m]}$. At this point, it seems plausible
to analyze the probabilistic behavior of this set of $m^{2}$ coefficients,
and complete the derivation of the bound. However, since, as noted
above, the bound involves a rather cumbersome counting argument over
the possible alternative sampling index vectors, a delicate analysis
seems intractable. Instead, we continue to further relax the upper
bound by using the maximal possible value $\overline{B}(X_{1}^{m})=\max_{i,j\in[m]\colon i\neq j}B(X_{i},X_{j})$.
While this may seem to be loose, we mention that alternative bounding
techniques that we have tried also lead to this maximal value. It
is also intuitively appealing, since whenever just a single pair $x_{i}$
and $x_{j}$ is close (in the Bhattacharyya distance) then clustering
error is likely, and the number of possible pairs is roughly $m^{2}$,
which is a constant that does not affect the error exponent. The analysis
of the contribution of this case to the error probability leads to
the expectation term in (\ref{eq: upper bound on clustering error finite ell}).
Now, while $\overline{B}(X_{1}^{m})$ has a typical value, it can
deviate from this value with probability exponential in $\ell$, and
thus not negligible. In what follows, the set ${\cal F}$ will be
used to control this deviation in the asymptotic analysis. In a similar
way, it will also be used to control deviations of the factor $\sqrt{\left(P_{X}^{\otimes\ell}(X_{1})\right)^{m-1}/\prod_{i=2}^{m}P_{X}^{\otimes\ell}(X_{i})}$
also appearing in (\ref{eq: upper bound on clustering error finite ell}),
which can affect the error exponent. This explains the origin of the
term $\P[X_{1}^{m}\in{\cal F}^{c}]$ in (\ref{eq: upper bound on clustering error finite ell}). 

\subsection{Error Exponent Lower Bound \label{subsec:Error-Exponent-Lower}}

Proposition \ref{prop: Raw upper bound on the error probability}
provides an explicit bound on the error probability, which holds for
any given source sequence length $\ell=\beta\log n$ sufficiently
large, yet it is still given as an expectation over ${\cal X}^{\otimes\ell}$
and $\ell=\beta\log n$ increases with $n$. We next evaluate the
error exponent. To this end, let $\tau>0$ be given, and consider
the set 
\begin{equation}
{\cal D}_{\tau}:=\left\{ x_{1}^{m}\in({\cal X}^{\otimes\ell})^{\otimes m}\colon\min_{i,\tilde{i}\in[m]\colon i\neq\tilde{i}}\frac{1}{\ell}D_{B}(x_{i},x_{\tilde{i}})\geq\tau\right\} .\label{eq: high probability set for a uniform source}
\end{equation}
Whenever $X_{1}^{m}\in{\cal D}_{\tau}$, it is guaranteed that the
Bhattacharyya distance between any pair of source sequences is larger
than $\tau\ell$. In order for the event $X_{1}^{m}\in{\cal D}_{\tau}^{c}$
to have probability exponentially small in $\ell$, the value of $\tau$
must be less than the expected value of the Bhattacharyya distance
$D_{B}(P_{X}\otimes P_{X})=\E_{X_{1},X_{2}\sim P_{X}^{\otimes2}}[D_{B}(X_{1},X_{2})]$.
For such $\tau$, since $D_{B}(X_{1},X_{2})$ is a sum of $\ell$
i.i.d. random variables $d_{B}(X_{1}(i),X_{2}(i))$, the method of
types \cite{csiszar1998method} (or Sanov's theorem) assures that
\[
\P\left[\frac{1}{\ell}D_{B}(X_{1},X_{2})<\tau\right]\leq e^{-E_{\text{B}}(\tau)\cdot\ell+o(\ell)}
\]
with exponent 
\begin{equation}
E_{\text{B}}(\tau):=\min_{Q_{X_{1}X_{2}}\in{\cal P}({\cal X}^{\otimes2})\colon D_{B}(Q_{X_{1}X_{2}})\leq\tau}D_{\text{KL}}(Q_{X_{1}X_{2}}\mid\mid P_{X}\otimes P_{X}).\label{eq: exponent of Bhatacharrya}
\end{equation}
Using Lagrange multipliers, it can be shown that the value of $E_{\text{B}}(\tau)$
can be computed by solving a single variable equation, see Lemma \ref{lem:Computation of E_b}
in Appendix \ref{sec:Useful-Lemmas}. Thus, by the union bound, the
probability of the event ${\cal D}_{\tau}^{c}$ follows the same exponential
decay as 
\begin{equation}
\P[X_{1}^{m}\in{\cal D}_{\tau}^{c}]\leq m^{2}\cdot e^{-E_{\text{B}}(\tau)\cdot\ell+o(\ell)}=e^{-E_{\text{B}}(\tau)\cdot\ell+o(\ell)}.\label{eq: large deviations for set F uniform}
\end{equation}
In a similar fashion, consider the event 
\begin{equation}
{\cal R}_{\rho}:=\left\{ x_{1}^{m}\in({\cal X}^{\otimes\ell})^{\otimes m}\colon\frac{1}{\ell}\log\frac{\left(P_{X}^{\otimes\ell}(X_{1})\right)^{m-1}}{\prod_{i=2}^{m}P_{X}^{\otimes\ell}(X_{i})}\leq\rho\right\} \label{eq: high probability source distributions}
\end{equation}
 for $\rho\in\mathbb{R}$. Note that 
\[
\E\left[\frac{1}{\ell}\log\frac{\left(P_{X}^{\otimes\ell}(X_{1})\right)^{m-1}}{\prod_{i=2}^{m}P_{X}^{\otimes\ell}(X_{i})}\right]=0
\]
and so ${\cal R}_{\rho}$ is a high probability event only when $\rho>0$.
Since $\log P_{X}^{\otimes\ell}(X_{i})=\sum_{j=1}^{\ell}\log P_{X}(X_{i}(j))$
for all $i\in[m]$, the method of types \cite{csiszar1998method}
(or Sanov's theorem) assures that
\[
\P[X_{1}^{m}\in{\cal R}_{\rho}^{c}]\leq e^{-E_{\text{P}}(\rho)\cdot\ell+o(\ell)},
\]
where 
\[
E_{\text{P}}(\rho):=\min_{Q_{X_{1}^{m}}\in{\cal P}({\cal X}^{\otimes m})\colon(m-1)\E_{X_{1}\sim Q_{X_{1}}}\left[\log P_{X}(X_{1})\right]+\sum_{i=2}^{m}\E_{X_{i}\sim Q_{X_{i}}}[\log P_{X}(X_{i})]\geq\rho}D_{\text{KL}}(Q_{X_{1}^{m}}\mid\mid P_{X}^{\otimes m}).
\]
Similarly, using Lagrange multipliers, it can be shown that the value
of $E_{\text{P}}(\rho)$ can be computed by solving a single variable
equation, see Lemma \ref{lem: Computation of Ep} in Appendix \ref{sec:Useful-Lemmas}.
We now choose ${\cal F}\equiv{\cal D}_{\tau}\cap{\cal R}_{\rho}$
and by the union bound
\begin{align}
\P[X_{1}^{m}\in{\cal F}^{c}] & \leq\P[X_{1}^{m}\in{\cal D}_{\tau}^{c}]+\P[X_{1}^{m}\in{\cal R}_{\rho}^{c}]\\
 & \leq e^{-E_{\text{B}}(\tau)\cdot\ell+o(\ell)}+e^{-E_{\text{P}}(\rho)\cdot\ell+o(\ell)}.
\end{align}
With this, we arrive at the following theorem.
\begin{thm}
\label{thm: Exponent for a uniform source}Assume that $d_{B,\text{\emph{max}}}<\infty$.
Then, 
\begin{equation}
\phi(\beta,P_{X},W)\geq\phi_{L}(\beta,P_{X},W):=\max_{\tau\geq0,\rho\geq0}E_{\text{\emph{B}}}(\tau)\wedge E_{\text{\emph{P}}}(\rho)\wedge\left[\tau-\rho-\frac{1}{\beta}\right].\label{eq: exponent lower bound general}
\end{equation}
Furthermore, when $P_{X}$ is the uniform distribution over ${\cal X}$,
if $E_{\text{\emph{B}}}(1/\beta)>0$ then 
\begin{equation}
\phi_{L}(\beta,P_{X},W)=\tau^{*}-\frac{1}{\beta}\label{eq: exponent lower bound uniform}
\end{equation}
where $\tau^{*}$ is the maximal solution to the equation $\tau-\frac{1}{\beta}=E_{\text{\emph{B}}}(\tau)$. 
\end{thm}
The proof of Theorem \ref{thm: Exponent for a uniform source} appears
in Appendix \ref{sec: proof of theorem exponent lower bound}. 

\paragraph*{Discussion}

The asymptotic bound of Theorem \ref{thm: Exponent for a uniform source}
does not depend on $m$, as long as it is fixed. For a uniform source,
that is, $P_{X}(a)=\frac{1}{|{\cal X}|}$ for all $a\in{\cal X}$,
it holds that $\sqrt{\left(P_{X}^{\otimes\ell}(X_{1})\right)^{m}/P_{X}^{\otimes\ell m}(X_{1}^{m})}=1$
with probability $1$. Thus, for any $\rho>0$ it holds that $E_{\text{P}}(\rho)=\infty$,
and the error exponent is obtained by balancing the exponent $E_{\text{B}}(\tau)$
of drawing source sequences $X_{1}^{m}$ with atypically low Bhattacharyya
distance, and $\tau-\frac{1}{\beta}$, the error exponent of random
errors in the sampling and reading processes. In the general case,
the exponent also accounts for atypical source sequences for which
$X_{1}$ has atypically high probability, and/or $X_{2}^{m}$ have
atypically low probability. As might be expected, it is evident from
the proof that the dominating error event is the one in which a single
read is assigned to a wrong cluster (specifically, this is suggested
by Lemma \ref{lem: binomial asymptotics}). Nonetheless, this is not
completely obvious, since the number of possible error events with
more than a single wrong read is much larger (though the probability
of erring to them is smaller). 

\subsection{Error Exponent Upper Bound \label{subsec: upper bound on the exponent}}

As suggested by the complexity of the proof of the lower bound on
the exponent, showing its tightness is challenging. Nonetheless, we
state here an upper bound on the exponent, which captures the dependence
on $\beta$, the source distribution, and the reading channel. The
key insight is to reduce the clustering problem to an assignment problem
of the reads to their originating source sequence. We then reduce
the latter problem to a multiple-hypothesis testing problem, which
is similar to the problem encountered in the random-coding analysis
of channel codes at zero rate \cite[Ch. 5]{gallager1968information},
\cite{merhav2025toolbox}. Before presenting the details of this reduction,
we state our bound. To this end, let 
\begin{equation}
H_{2}(Q):=-\log\left(\sum_{x\in{\cal X}}Q_{X}^{2}(x)\right)\label{eq: Second order Renyi entropy}
\end{equation}
denote the second-order R\'{e}nyi entropy. Let Gallager's random-coding
error exponent at rate $R$ be \cite[Ch. 5]{gallager1968information}
\[
E_{r}(P_{X},R)=\max_{\rho\in[0,1]}-\left\{ \log\sum_{y\in{\cal Y}}\left(\sum_{x\in{\cal X}}P_{X}(x)W^{1/(1+\rho)}(y\mid x)\right)^{1+\rho}-\rho R\right\} .
\]
Specifically, we will use this bound for zero rate $R=0$ for which
we can choose $\rho=1$ (which is known to be optimal at low rates),
to obtain
\begin{align}
E_{r}(P_{X},0) & =-\log\sum_{y\in{\cal Y}}\left(\sum_{x\in{\cal X}}P_{X}(x)\sqrt{W(y\mid x)}\right)^{2}\\
 & =-\log\sum_{y\in{\cal Y}}\sum_{x\in{\cal X}}P_{X}(x)\sqrt{W(y\mid x)}\sum_{\tilde{x}\in{\cal X}}P_{X}(\tilde{x})\sqrt{W(y\mid\tilde{x})}\\
 & =-\log\sum_{x\in{\cal X}}\sum_{\tilde{x}\in{\cal X}}P_{X}(x)P_{X}(\tilde{x})\sum_{y\in{\cal Y}}\sqrt{W(y\mid x)}\sqrt{W(y\mid\tilde{x})}\\
 & =-\log\E_{(X_{1},X_{2})\sim P_{X}\times P_{X}}\left[B(X_{1},X_{2})\right].
\end{align}
So, with a slight abuse of notation, we denote 
\[
B(Q_{X_{1}X_{2}}):=\E_{Q_{X_{1}X_{2}}}\left[B(X_{1},X_{2})\right].
\]
Note that from Jensen's inequality, 
\begin{align}
E_{r}(P_{X},0) & =-\log B(P_{X}\times P_{X})\\
 & =-\log\E_{(X_{1},X_{2})\sim P_{X}\times P_{X}}\left[B(X_{1},X_{2})\right]\\
 & \leq\E_{(X_{1},X_{2})\sim P_{X}\times P_{X}}\left[-\log B(X_{1},X_{2})\right]\\
 & =D_{B}(P_{X}\times P_{X}),
\end{align}
using the definition of the average Bhattacharyya distance in (\ref{eq: average Bhattacharyya coefficient}). 
\begin{thm}
\label{thm: upper bound on the exponent}Assume w.l.o.g. that $d_{B,\text{min}}>0$,
and further assume that 
\begin{equation}
E_{\text{\emph{B}}}(1/\beta)>-\log B(P_{X}\times P_{X}),\label{eq: upper bound first condition}
\end{equation}
and 
\begin{equation}
H_{2}(P_{X})>-\log B(P_{X}\times P_{X})-\frac{1}{\beta}.\label{eq: upper bound second condition}
\end{equation}
Then, 
\begin{equation}
\phi(\beta,P_{X},W)\leq\phi_{U}(\beta,P_{X},W):=-\log B(P_{X}\times P_{X})-\frac{1}{\beta}.\label{eq: exponent upper bound general}
\end{equation}
\end{thm}
The lower and upper bounds in Theorems \ref{thm: Exponent for a uniform source}
and \ref{thm: upper bound on the exponent} have a similar $-1/\beta$
dependence, and the term $-\log B(P_{X}\times P_{X})$ depends on
both the source sequence distribution $P_{X}$ and the reading channel
$W$. 

\paragraph{Main ideas of the proof}

As noted, our proof strategy is to reduce the clustering problem to
an assignment problem. The goal in the assignment problem is to assign
each read $y_{j}$ for $j\in[n]$ to the source sequence that generated
it, that is, to correctly identify the index $S_{j}$, given that
the source sequences are known. Formally, an assignment rule is $\mathsf{A}\colon({\cal X}^{\otimes\ell})^{\otimes m}\otimes({\cal Y}^{\otimes\ell})^{\otimes n}\to[m]^{\otimes n}$,
and its error probability is given by
\begin{equation}
\tilde{p}_{\text{error}}(\mathsf{A})=\P\left[\mathsf{A}(X_{1}^{m},Y_{1}^{n})\neq S_{1}^{n}\right].\label{eq: error probability definition assignment}
\end{equation}
In this problem, too, the optimal assignment rule is given by the
MAP rule $\mathsf{A}_{\text{MAP}}$. At first glance, the assignment
problem is strictly easier  than the clustering problem since the
realizations of the source sequences $X_{1}^{m}$ are known in the
former. However, the somewhat delicate issue is that the clustering
rule is invariant to permutations of the labels without making an
error, whereas the assignment problem is not.\footnote{The notation $\tilde{p}_{\text{error}}$ in (\ref{eq: error probability definition assignment})
emphasizes this difference from the error probability definition for
clustering algorithms in (\ref{eq: error probability definition}).} For example, if $n=5$ and $S_{1}^{n}=(1,2,1,1,2)$ then $\mathsf{C}_{\text{MAP}}(y_{1}^{n})=(2,1,2,2,1)$
is not a clustering error, but is an assignment error. The next proposition
shows that this issue has a negligible effect in our regime.
\begin{prop}
\label{prop: asignment versus clustering error probability}Assume
that $H_{2}(P_{X})>0$ and assume w.l.o.g. that $d_{B,\text{min}}>0$.
Further assume that $m$ is fixed and $\ell=\beta\log n$. Then, 
\[
p_{\text{\emph{error}}}(\mathsf{C}_{\text{\emph{MAP}}})\ge\tilde{p}_{\text{\emph{error}}}(\mathsf{A}_{\text{\emph{MAP}}})-e^{-\ell H_{2}(P_{X})+o(\ell)}.
\]
\end{prop}
The proof of Proposition \ref{prop: asignment versus clustering error probability}
appears in Appendix \ref{sec: proof of theorem exponent upper bound}.
Given Proposition \ref{prop: asignment versus clustering error probability},
it remains to lower bound $\tilde{p}_{\text{error}}(\mathsf{A}_{\text{MAP}})$.
To achieve this, we first condition on $X_{1}^{m}$, and exploit the
fact that the reads $Y_{1}^{n}$ are then conditionally independent.
Thus, conditioned on $X_{1}^{m}$, the optimal assignment is performed
separately for each read, and the resulting error events are pairwise
independent. In this setting, the \emph{clipped} union bound over
the $n$ error events (one for each read) is tight. If this were a
standard union bound, then the assignment error probability would
be $n$ times the error probability of a single read, and then averaging
over the randomness of $X_{1}^{m}$ would immediately lead to the
random coding exponent at rate $\frac{\log m}{\ell}=\frac{\log m}{\beta\log n}=o(1)$,
that is, at zero rate. The main technical difficulty is the need to
use a \emph{clipped} union bound. This is addressed in the detailed
proof of Theorem \ref{thm: upper bound on the exponent}, which also
appears in Appendix \ref{sec: proof of theorem exponent upper bound}.

\subsection{Illustrative Examples}

We next illustrate the bounds of Theorems \ref{thm: Exponent for a uniform source}
and \ref{thm: upper bound on the exponent} in simple numerical examples,
involving the binary symmetric channel (BSC) and its $q$-ary symmetric
extension. 

\paragraph*{Binary alphabets observed through a BSC}

Figure \ref{fig:bsc-bounds} displays the exponent lower bound and
upper bound for a BSC with crossover probabilities $p\in\{0.05,0.10,0.20\}$,
assuming that $P_{X}$ is uniform over ${\cal X}=\{0,1\}$. 
\begin{figure}[H]
\centering \includegraphics[scale=0.6]{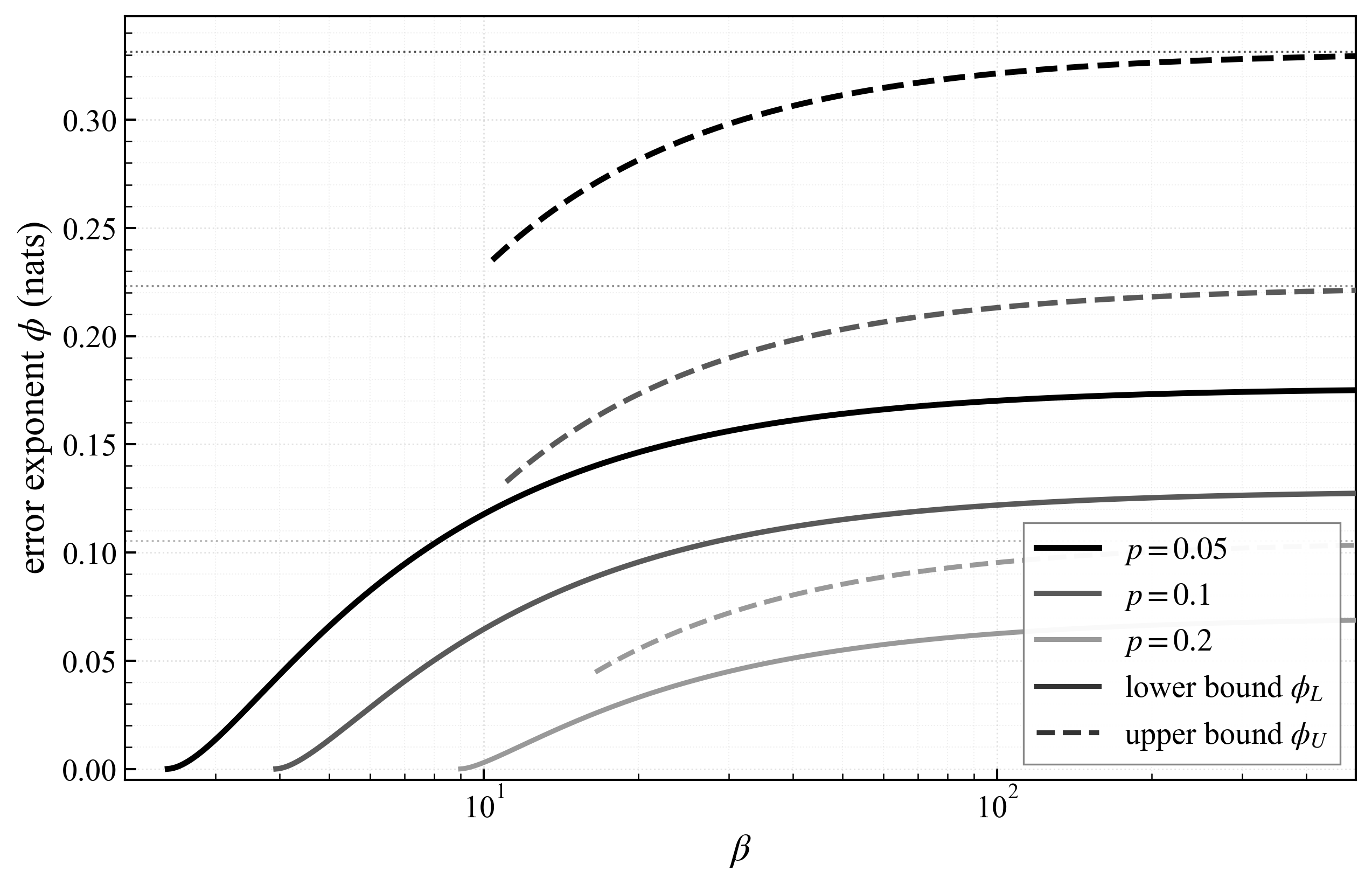} \caption{BSC, $p\in\{0.05,0.10,0.20\}$, uniform binary $P_{X}$. Solid: Lower
bound $\phi_{L}$ (Theorem \ref{thm: Exponent for a uniform source}).
Dashed: Upper bound $\phi_{U}$ (Theorem \ref{thm: upper bound on the exponent}).
Dotted: $-\log B(P_{X}\times P_{X})$ at each $p$.}
\label{fig:bsc-bounds}
\end{figure}
 Theorem \ref{thm: Exponent for a uniform source} yields a positive
bound only when $E_{\text{B}}(1/\beta)>0$, i.e. $\beta>\beta_{3}(p):=1/D_{B}(P_{X}\times P_{X})$,
equal to $2/d_{B}(0,1)$ for the BSC. The condition (\ref{eq: upper bound first condition})
of Theorem \ref{thm: upper bound on the exponent} $E_{\text{B}}(1/\beta)>-\log B(P_{X}\times P_{X})$,
defines a second threshold $\beta_{4}(p)$. Figure \ref{fig:thresholds}
traces both and partitions the $(\beta,p)$ plane into three regions:
one where neither bound applies, one where only the lower bound applies,
and one where both apply. 
\begin{figure}[H]
\centering \includegraphics[scale=0.6]{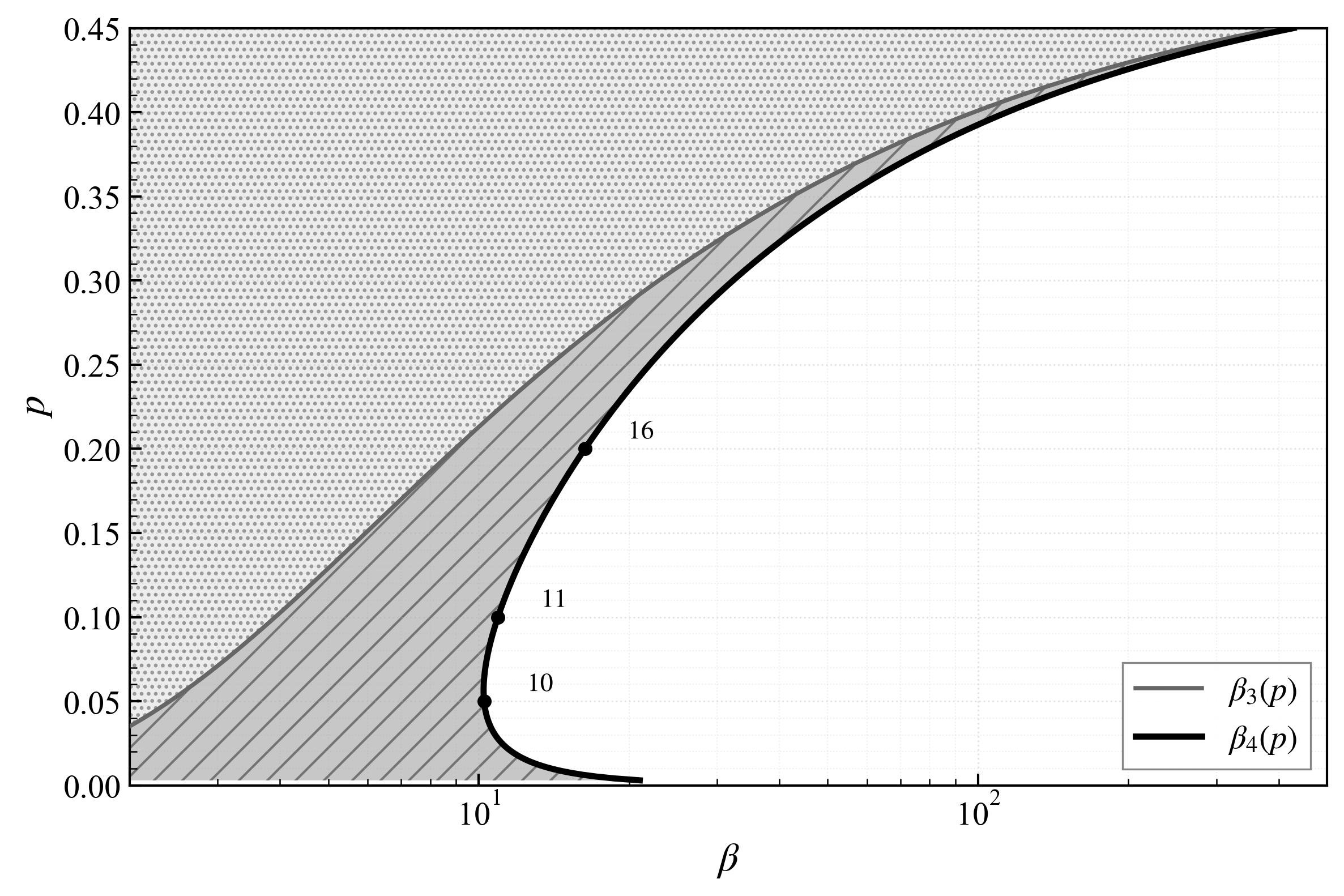} \caption{Thresholds in the $(\beta,p)$ plane for the BSC, uniform $P_{X}$.
Thin: Theorem \ref{thm: Exponent for a uniform source} threshold
$\beta_{3}(p)$ (the lower bound is zero for $\beta<\beta_{3}$).
Thick: Theorem \ref{thm: upper bound on the exponent} threshold $\beta_{4}(p)$
(the curve $E_{\text{B}}(1/\beta)=-\log B(P_{X}\times P_{X})$). Shading
distinguishes the three regions of the plane. Values of $\beta_{4}$
at $p\in\{0.05,0.10,0.20\}$ are annotated.}
\label{fig:thresholds}
\end{figure}

\paragraph*{Larger alphabets observed through a symmetric channel}

We next consider the effect of the alphabet size, assuming a $q$-ary
symmetric channel in which 
\begin{equation}
W_{q}(y\mid x)=\begin{cases}
1-p, & y=x\\
\frac{p}{q-1} & \text{otherwise}
\end{cases},\label{eq: symmetric channel}
\end{equation}
and $P_{X}$ is uniform over ${\cal X}=\{0,1,\ldots,q-1\}$. Figure
\ref{fig:q-compare} compares the lower bound and upper bound for
$q\in\{2,3,4\}$, across three crossover probabilities. The exponents
grow with $q$ as distinct symbols become easier to distinguish, but
the ratio between the bound $\phi_{U}/\phi_{L}$ stays roughly near
$1.8$ regardless of $q$. 
\begin{figure}[H]
\centering \includegraphics[scale=0.4]{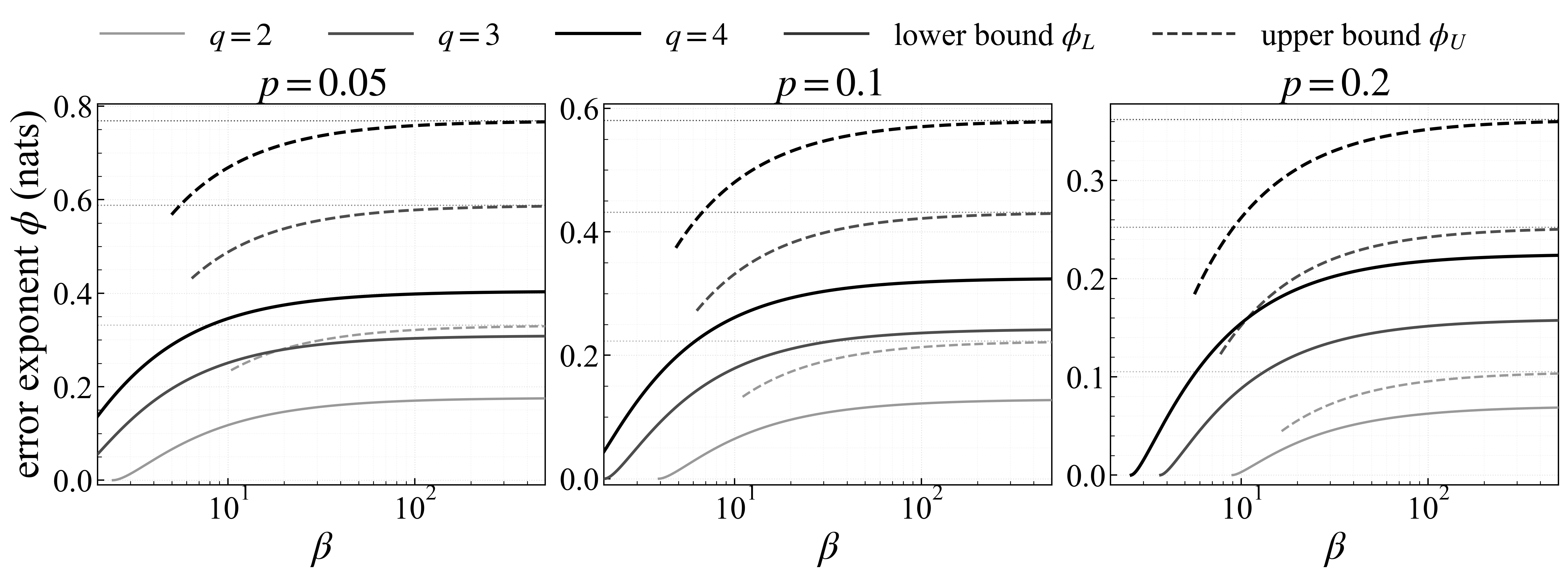} \caption{$q$-ary symmetric channel, $q\in\{2,3,4\}$, $p\in\{0.05,0.10,0.20\}$,
uniform $P_{X}$. Solid: Lower bound $\phi_{L}$ (Theorem \ref{thm: Exponent for a uniform source}).
Dashed: Upper bound $\phi_{U}$ (Theorem \ref{thm: upper bound on the exponent}).
Shade and width encode $q$ (darker and thicker is for larger $q$).
Dotted: $-\log B(P_{X}\times P_{X})$ for each $q$.}
\label{fig:q-compare}
\end{figure}
 The thresholds extend to the $q$-ary case. Figure \ref{fig:q-regions}
displays $\beta_{3}(p,q)$ and $\beta_{4}(p,q)$ for $q\in\{2,3,4\}$.
The band between $\beta_{3}$ and $\beta_{4}$, where only the lower
bound applies, narrows as $q$ grows, so the regime where both bounds
apply grows with the alphabet size. 
\begin{figure}[H]
\centering \includegraphics[scale=0.6]{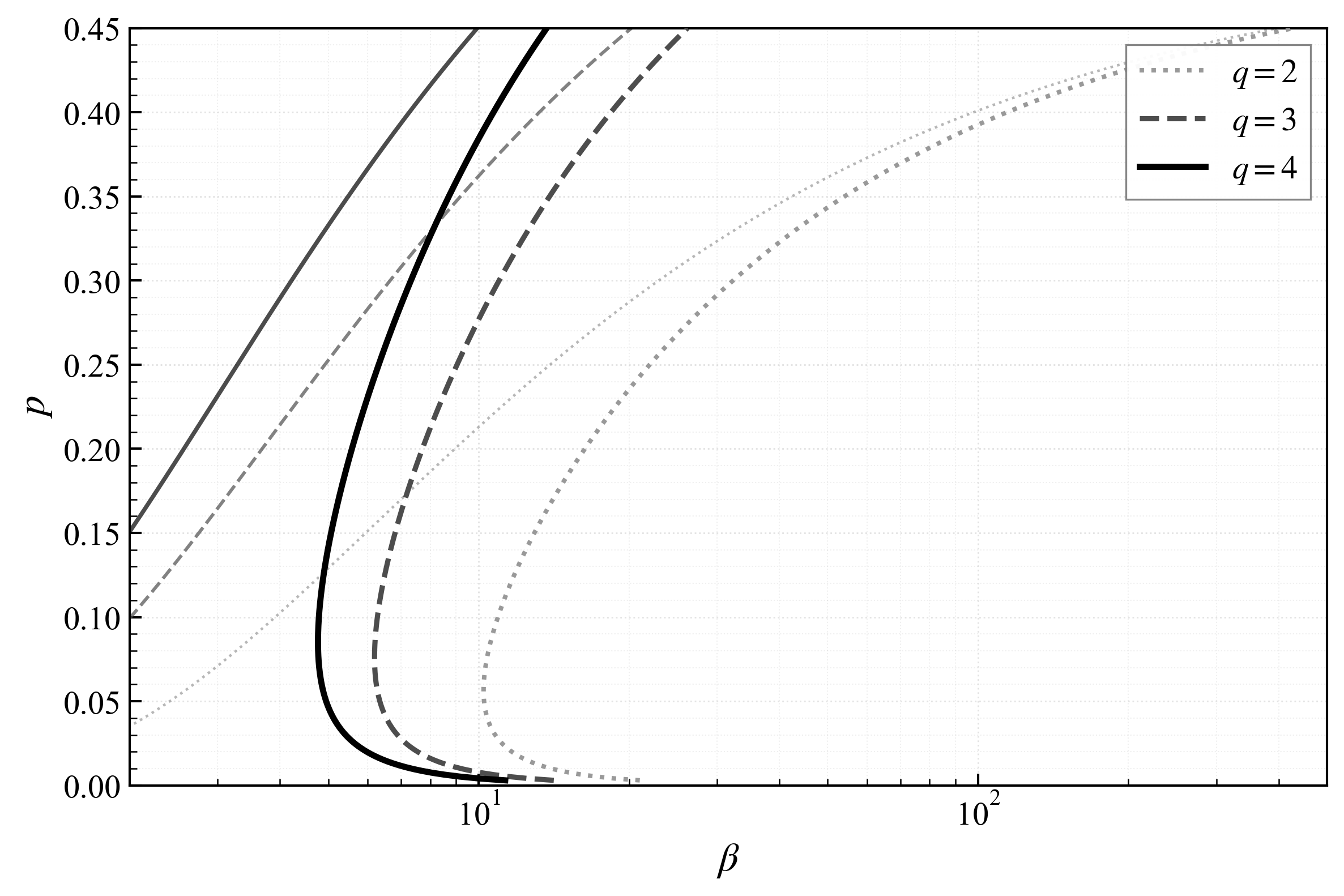} \caption{Thresholds in the $(\beta,p)$ plane for the $q$-ary symmetric channel,
$q\in\{2,3,4\}$, uniform $P_{X}$. Thick: $\beta_{4}(p,q)$. Thin:
$\beta_{3}(p,q)$. The line style encodes $q$: dotted ($q=2$), dashed
($q=3$), solid ($q=4$, the DNA alphabet). }
\label{fig:q-regions}
\end{figure}

\paragraph*{Non-uniform sequence distribution observed through a BSC}

Figure \ref{fig:nonuniform} displays the error exponent bounds for
a BSC with crossover probability $p=0.1$, both for the uniform source
and two biased binary sources $P_{X}=(q_{p},1-q_{p})$. It can be
seen that the lower bound becomes loose as $q_{p}$ moves away from
$\tfrac{1}{2}$, while the upper bound is far less affected. Sharpening
the bounds for unbalanced sources remains an open challenge. 
\begin{figure}[H]
\centering \includegraphics[scale=0.23]{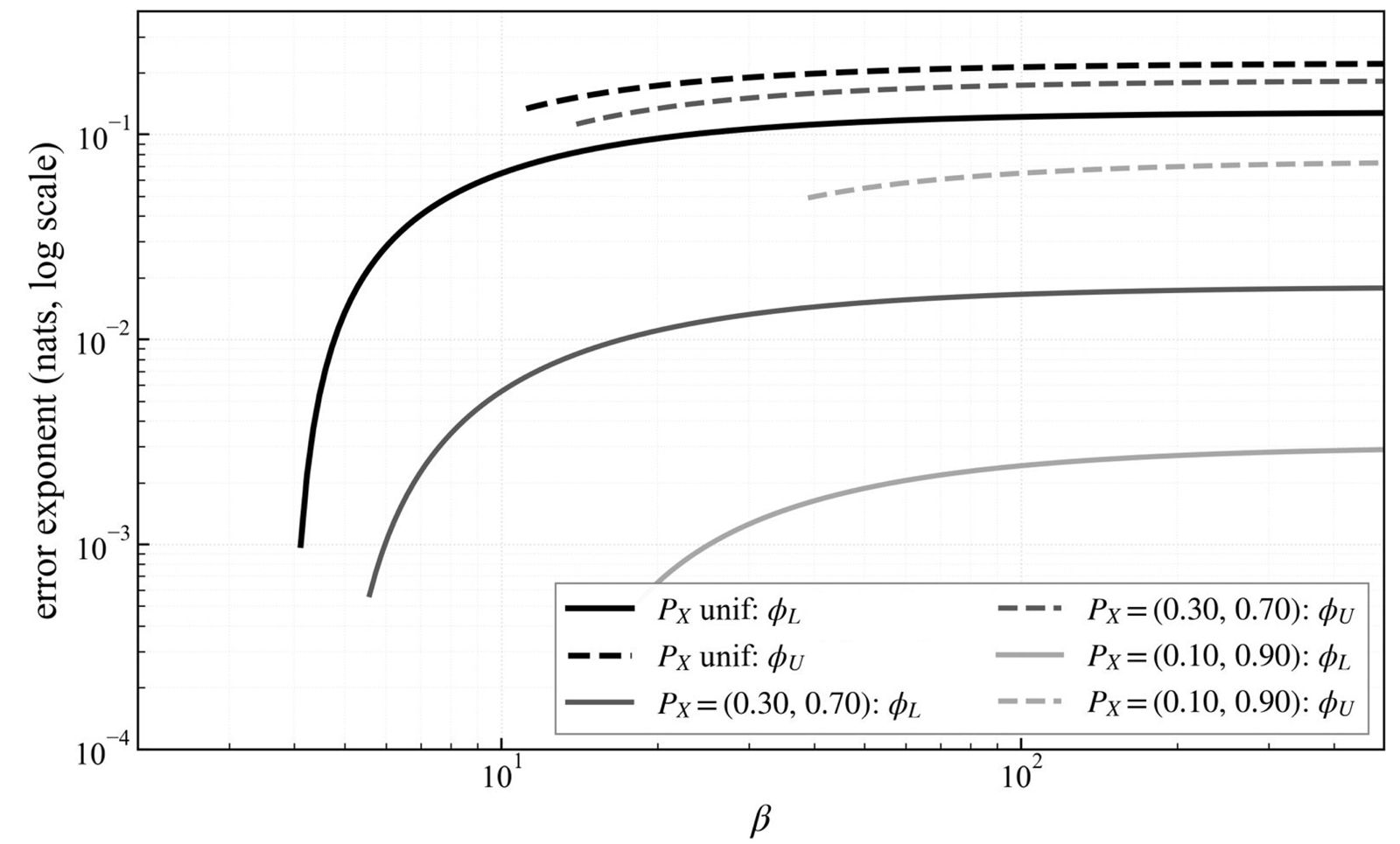} \caption{BSC with crossover probability $p=0.1$, under uniform and two biased
binary $P_{X}$. Solid: The lower bound $\phi_{L}$ (Theorem \ref{thm: Exponent for a uniform source}).
Dashed: The upper bound $\phi_{U}$ (Theorem \ref{thm: upper bound on the exponent}).
Black: uniform $P_{X}=(\frac{1}{2},\frac{1}{2})$. Grey: $P_{X}=(\frac{3}{10},\frac{7}{10})$
and $(\frac{1}{10},\frac{9}{10})$. Each curve is plotted only over
the range where its lower bound is non-zero or where the condition
for the upper bound holds.}
\label{fig:nonuniform}
\end{figure}

\section{Proof of Proposition \ref{prop: Raw upper bound on the error probability}
\label{sec:Proof-of- upper bound on error probability}}

Let $R_{1}^{m}\in\{0,1,2,\ldots,n\}{}^{\otimes m}$ be the histogram
of the sampling index vector $S_{1}^{n}$, that is, for $i\in[m]$,
it holds that 
\[
R_{i}=:\sum_{j=1}^{n}\I[S_{j}=i]
\]
is the number of times that the $i$th source sequence has been sampled.
So $\sum_{i=1}^{m}R_{i}=n$, and clearly, the expected number of times
each source sequence is sampled is $\E[R_{i}]=\frac{n}{m}$. Our first
lemma shows that $R_{1}^{m}$ concentrates rapidly around $(\frac{n}{m},\ldots,\frac{n}{m})$.
To this end, let $\eta\in(0,1/4)$ be given and let us define the
set
\begin{equation}
{\cal G}_{\eta}:=\left\{ r_{1}^{m}\in\mathbb{N}_{+}^{\otimes m}\colon\left\{ \left|r_{i}-\frac{n}{m}\right|\leq\eta\cdot\frac{n}{m}\right\} \text{for all }i\in[m]\right\} .\label{eq: concentration of sampling histogram event}
\end{equation}
By Hoeffding's inequality and the union bound, it holds that 
\begin{equation}
\P\left[R_{1}^{m}\not\in{\cal G}_{\eta}\right]=\P\left[\bigcup_{i=1}^{m}\left\{ \left|R_{i}-\frac{n}{m}\right|>\eta\cdot\frac{n}{m}\right\} \right]\leq me^{-c\eta^{2}n}\label{eq: concentration of number of times each source is sampled}
\end{equation}
for some numerical constant $c>0$. In a similar fashion, we also
define another set ${\cal F}\subseteq({\cal X}^{\otimes\ell})^{\otimes m}$,
which is arbitrary (eventually we will choose ${\cal F}$ to be a
high probability set, that is $\P[X_{1}^{m}\in{\cal F}]\to1$). Since
$R_{1}^{m}\in{\cal G}_{\eta}$ with high probability, we will analyze
the (possibly) suboptimal clustering rule in which the maximization
over $s_{1}^{n}$ is only over ${\cal G}_{\eta}$, that is, following
(\ref{eq: approximate ML rule second}), we will analyze the clustering
rule
\begin{equation}
\mathsf{C}_{\text{AML}}^{\sharp}(y_{1}^{n})=\argmax_{s_{1}^{n}\in[m]^{\otimes n}\colon r_{1}^{m}\in{\cal G}_{\eta}}\prod_{i=1}^{m}\sum_{x\in{\cal X}}P_{X}(x)\prod_{j\in[n]\colon s_{j}=i}W(y_{j}\mid x).\label{eq: approximate clustering for proof}
\end{equation}
Denoting the clustering error event of this rule by (see (\ref{eq: error probability definition}))
\[
{\cal E}:=\bigcap_{\pi\in\Pi_{m}}\left\{ \mathsf{C}_{\text{AML}}^{\sharp}(Y_{1}^{n})\neq\pi(S_{1}^{n})\right\} 
\]
we may upper bound the error probability of the MAP clustering rule
as
\begin{align}
p_{\text{error}}(\mathsf{C}_{\text{MAP}}) & \trre[\leq,a]p_{\text{error}}(\mathsf{C}_{\text{AML}}^{\sharp})\\
 & \trre[=,b]\E\left[\I\left\{ {\cal E}\cap{\cal G}_{\eta}\cap{\cal F}\right\} \right]+\E\left[\I\left\{ {\cal E}\cap({\cal G}_{\eta}^{c}\cup{\cal F}^{c})\right\} \right]\\
 & \trre[\leq,c]\E\left[\I\left\{ {\cal E}\cap{\cal G}_{\eta}\cap{\cal F}\right\} \right]+\P\left[R_{1}^{m}\not\in{\cal G}_{\eta}\right]+\P\left[X_{1}^{m}\not\in{\cal F}\right]\\
 & \trre[\leq,d]\E\left[\I\left\{ {\cal E}\cap{\cal G}_{\eta}\cap{\cal F}\right\} \right]+me^{-c\eta^{2}n}+\P\left[X_{1}^{m}\not\in{\cal F}\right]\\
 & =\sum_{r_{1}^{m}\in{\cal G}_{\eta}}\P\left[R_{1}^{m}=r_{1}^{m}\right]\P\left[{\cal E}\cap{\cal F}\mid R_{1}^{m}=r_{1}^{m}\right]+me^{-c\eta^{2}n}+\P\left[X_{1}^{m}\not\in{\cal F}\right]\\
 & \leq\max_{r_{1}^{m}\in{\cal G}_{\eta}}\P\left[{\cal E}\cap{\cal F}\mid R_{1}^{m}=r_{1}^{m}\right]+me^{-c\eta^{2}n}+\P\left[X_{1}^{m}\not\in{\cal F}\right],\label{eq: main decomposition of error probability}
\end{align}
where $(a)$ follows since $\mathsf{C}_{\text{AML}}^{\sharp}$ is
a suboptimal clustering rule, $(b)$ uses the shorthand notation ${\cal G}_{\eta}\equiv\{R_{1}^{m}\in{\cal G}_{\eta}\}$
and ${\cal F}\equiv\{X_{1}^{m}\in{\cal F}\}$, $(c)$ follows from
the union bound, and $(d)$ follows from (\ref{eq: concentration of number of times each source is sampled}).
In what follows, we thus consider an arbitrary $r_{1}^{m}\in{\cal G}_{\eta}$
and upper bound the conditional error probability. Due to the permutation
symmetry, we may assume that the sampling index vector is, w.l.o.g.,
the \emph{canonical} $s_{1}^{n}=(\underbrace{1,1,\ldots,1}_{r_{1}\,\text{times}},\underbrace{2,2,\ldots,2}_{r_{2}\,\text{times}},\ldots,\underbrace{m,m,\ldots,m}_{r_{m}\,\text{times}})$.
We thus bound $\P[{\cal E}\mid S_{1}^{n}=s_{1}^{n}]$ for this canonical
$s_{1}^{n}$. Now, the union bound implies that the conditional error
probability is upper bounded by a union over all possible alternative
sampling vectors, with different pattern. Due to the choice of the
rule (\ref{eq: approximate clustering for proof}), we only consider
$\tilde{s}_{1}^{n}$ such that $\tilde{r}_{1}^{m}\in{\cal G}_{\eta}$,
where $\tilde{r}_{1}^{m}$ is defined analogously to $r_{1}^{m}$,
as
\[
\tilde{r}_{i}=\sum_{j=1}^{n}\I[\tilde{s}_{j}=i].
\]
By the union bound, 
\begin{equation}
\P[{\cal E}\cap{\cal F}\mid S_{1}^{n}=s_{1}^{n}]\leq\sum_{\tilde{s}_{1}^{n}\in\Pi(n,m,m)\colon\tilde{r}_{1}^{m}\in{\cal G}_{\eta},\Psi(s_{1}^{n})\neq\Psi(\tilde{s}_{1}^{n})}\P\left[\lambda(Y_{1}^{n}\mid\tilde{s}_{1}^{n})\geq\lambda(Y_{1}^{n}\mid s_{1}^{n})\mid S_{1}^{n}=s_{1}^{n}\right],\label{eq: pairwise union bound}
\end{equation}
where we recall that $\Psi(s_{1}^{n})$ is the pattern of $s_{1}^{n}$,
and we sum over all patterns which have exactly $m$ different letters.
We next analyze the pairwise error probability, from $s_{1}^{n}$
to an alternative sampling index vector $\tilde{s}_{1}^{n}$, where
$\Psi(s_{1}^{n})\neq\Psi(\tilde{s}_{1}^{n})$. For $i,k\in[m]$ let
\[
t_{i\to k}:=\sum_{j=1}^{n}\I[s_{j}=i,\tilde{s}_{j}=k]
\]
denote the number of indices in which the sampling index of $s_{1}^{n}$
is $i$ and in $\tilde{s}_{1}^{n}$ is $k$. Thus, $\overline{\overline{t}}\equiv\overline{\overline{t}}(s_{1}^{n},\tilde{s}_{1}^{n}):=(t_{i\to k}){}_{(i,k)\in[m]^{\otimes2}}$
is the joint histogram of $(s_{1}^{n},\tilde{s}_{1}^{n})$. We can
consider it a matrix, and for brevity, we will denote each one of
its rows by $\overline{t}_{i\to}:=(\overline{t}_{i\to1},\overline{t}_{i\to2},\ldots,\overline{t}_{i\to m})$.
We note that it holds that 
\[
r_{i}=\sum_{k=1}^{m}t_{i\to k},
\]
for all $i\in[m]$ and
\[
\tilde{r}_{i}=\sum_{i=1}^{m}t_{k\to i},
\]
for all $i\in[m]$. To further upper bound the error probability,
we need to introduce a few definitions. First, let ${\cal S}_{n}\equiv{\cal S}_{n}(s_{1}^{n},\overline{\overline{t}})$
denote the set of alternative sampling index vectors that are obtained
from the canonical $s_{1}^{n}$ by $\overline{\overline{t}}$. Using
the permutation symmetry of the reads, the conditional pairwise error
from $s_{1}^{n}$ to $\tilde{s}_{1}^{n}$ is the same for all $\tilde{s}_{1}^{n}\in{\cal S}_{n}$.
Second, let 
\[
{\cal J}_{m}(r):=\left\{ \overline{q}=(q_{1},q_{2},\ldots,q_{m})\in\mathbb{N}_{+}^{\otimes m}\colon\sum_{i=1}^{m}q_{i}=r\right\} 
\]
denote the set of tuples for which the multinomial coefficient $\text{\ensuremath{\binom{r}{q_{1},\ldots,q_{m}}}\ensuremath{\ensuremath{\equiv}}\ensuremath{\binom{r}{\overline{q}}}}$
is non-zero. Third, we note that if $t_{i\to i}=r_{i}$ for all $i\in[m]$
then $s_{1}^{n}=\tilde{s}_{1}^{n}$ and this is not a clustering error.
Moreover, if there exists a permutation $\pi\in\Pi_{m}$ (the symmetric
group of $[m]$) such that $t_{i\to\pi(i)}=r_{i}$ for all $i\in[m]$
then $\Psi(s_{1}^{n})=\Psi(\tilde{s}_{1}^{n})$ (i.e., $s_{1}^{n}$
and $\tilde{s}_{1}^{n}$ have the same pattern) and this is also not
a clustering error. We thus define the set 
\[
{\cal T}_{0}:=\left\{ \overline{\overline{t}}\colon\text{There exists }\pi\in\Pi_{m}\text{ such that }t_{i\to\pi(i)}=r_{i}\text{ for all }i\in[m]\right\} .
\]
As noted, if $\overline{\overline{t}}\equiv\overline{\overline{t}}(s_{1}^{n},\tilde{s}_{1}^{n})\in{\cal T}_{0}$
then deciding on $\tilde{s}_{1}^{n}$ is not a clustering error. With
these definitions, we may further bound (\ref{eq: pairwise union bound})
as 
\begin{align}
 & \P[{\cal E}\cap{\cal F}\mid S_{1}^{n}=s_{1}^{n}]\nonumber \\
 & \leq\sum_{\overline{t}_{1\to}\in{\cal J}(r_{1})}\sum_{\overline{t}_{2\to}\in{\cal J}(r_{2})}\cdots\sum_{\overline{t}_{m\to}\in{\cal J}(r_{m})}\I\left\{ \overline{\overline{t}}(s_{1}^{n},\tilde{s}_{1}^{n})\not\in{\cal T}_{0}\right\} \cdot\ensuremath{\binom{r_{1}}{\overline{t}_{1\to}}}\ensuremath{\binom{r_{2}}{\overline{t}_{2\to}}}\cdots\ensuremath{\binom{r_{m}}{\overline{t}_{m\to}}}\nonumber \\
 & \hphantom{=======}\times\P\left[\lambda(Y_{1}^{n}\mid\tilde{s}_{1}^{n})\geq\lambda(Y_{1}^{n}\mid s_{1}^{n}),{\cal F}\mid S_{1}^{n}=s_{1}^{n},\tilde{s}_{1}^{n}\in{\cal S}_{n}(s_{1}^{n},\overline{\overline{t}}),\;\tilde{r}_{1}^{m}\in{\cal G}_{\eta}\right].\label{eq:  main bound on the error probability fixed x1 and x2}
\end{align}
Note that this is a generous bound, since some of the choices for
$\overline{\overline{t}}$ counted here do not correspond to $\tilde{r}_{1}^{m}\in{\cal G}_{\eta}$. 

Our next lemma is a bound on the pairwise error probability, using
a Bhattacharyya-style upper bound. 
\begin{lem}
\label{lem: Bhat upper bound on the pairwise error probability}Let
$s_{1}^{n}$ and $\overline{\overline{t}}=(t_{j\to k}){}_{(j,k)\in[m]^{\otimes2}}$
be given and let $\tilde{s}_{1}^{n}\in{\cal S}_{n}(s_{1}^{n},\overline{\overline{t}})$.
Let 
\begin{equation}
f_{i}(x_{1}^{m}):=\sum_{\tilde{x}\in{\cal X}^{\otimes\ell}}\sqrt{P_{X}^{\otimes\ell}(\tilde{x})}\prod_{k=1}^{m}B^{t_{k\to i}}(x_{k},\tilde{x}).\label{eq: definition of F_i}
\end{equation}
Then, for any ${\cal F}\subseteq({\cal X}^{\otimes\ell})^{\otimes m}$
\begin{align}
 & \P\left[\lambda(Y_{1}^{n}\mid\tilde{s}_{1}^{n})\geq\lambda(Y_{1}^{n}\mid s_{1}^{n}),{\cal F}\mid S_{1}^{n}=s_{1}^{n}\right]\nonumber \\
 & \leq\sum_{x_{1}^{m}\in{\cal F}}\sqrt{P_{X}^{\otimes\ell m}(x_{1}^{m})}\cdot\prod_{i=1}^{m}f_{i}(x_{1}^{m}).\label{eq: Bhat upper bound on pairwise error probability statement}
\end{align}
\end{lem}
\begin{IEEEproof}
The pairwise error probability is upper bounded by a Bhattacharyya-bound
argument as
\begin{align}
 & \P\left[\lambda(Y_{1}^{n}\mid\tilde{s}_{1}^{n})\geq\lambda(Y_{1}^{n}\mid s_{1}^{n}),{\cal F}\mid S_{1}^{n}=s_{1}^{n}\right]\nonumber \\
 & \leq\E\left[\I\left\{ \lambda(Y_{1}^{n}\mid\tilde{s}_{1}^{n})\geq\lambda(Y_{1}^{n}\mid s_{1}^{n})\right\} \cdot\I\{X_{1}^{m}\in{\cal F}\}\mid S_{1}^{n}=s_{1}^{n}\right]\\
 & \leq\E\left[\sqrt{\frac{\lambda(Y_{1}^{n}\mid\tilde{s}_{1}^{n})}{\lambda(Y_{1}^{n}\mid s_{1}^{n})}}\cdot\I\{X_{1}^{m}\in{\cal F}\}\mid S_{1}^{n}=s_{1}^{n}\right]\\
 & =\sum_{x_{1}^{m}\in({\cal X}^{\otimes\ell})^{\otimes m}}\sum_{y_{1}^{n}\in({\cal Y}^{\otimes\ell})^{\otimes n}}P_{X}^{\otimes\ell m}(x_{1}^{m})\prod_{i=1}^{m}\prod_{j\in[n]\colon s_{j}=i}W^{(\ell)}(y_{j}\mid x_{i})\I\{X_{1}^{m}\in{\cal F}\}\times\nonumber \\
 & \hphantom{==============}\sqrt{\frac{\sum_{x_{1}^{m}\in({\cal X}^{\otimes\ell})^{\otimes m}}P_{X}^{\otimes\ell m}(x_{1}^{m})\prod_{i=1}^{m}\prod_{j\in[n]\colon s_{j}=i}W^{(\ell)}(y_{j}\mid x_{i})}{\sum_{x_{1}^{m}\in({\cal X}^{\otimes\ell})^{\otimes m}}P_{X}^{\otimes\ell m}(x_{1}^{m})\prod_{i=1}^{m}\prod_{j\in[n]\colon s_{j}=i}W^{(\ell)}(y_{j}\mid x_{i})}}\\
 & \trre[\leq,a]\sqrt{\sum_{x_{1}^{m}\in{\cal F}}\sum_{y_{1}^{n}\in({\cal Y}^{\otimes\ell})^{\otimes n}}P_{X}^{\otimes\ell m}(x_{1}^{m})\prod_{i=1}^{m}\prod_{j\in[n]\colon s_{j}=i}W^{(\ell)}(y_{j}\mid x_{i})}\times\nonumber \\
 & \hphantom{==============}\sqrt{\sum_{x_{1}^{m}\in({\cal X}^{\otimes\ell})^{\otimes m}}P_{X}^{\otimes\ell m}(x_{1}^{m})\prod_{i=1}^{m}\prod_{j\in[n]\colon s_{j}=i}W^{(\ell)}(y_{j}\mid x_{i})}\\
 & \trre[\leq,b]\sum_{y_{1}^{n}\in({\cal Y}^{\otimes\ell})^{\otimes n}}\sum_{x_{1}^{m}\in{\cal F}}\sum_{\tilde{x}_{1}^{m}\in({\cal X}^{\otimes\ell})^{\otimes m}}\sqrt{\left(P_{X}^{\otimes\ell m}(x_{1}^{m})\prod_{j=1}^{n}W^{(\ell)}(y_{j}\mid x_{s_{j}})\right)\times\left(P_{X}^{\otimes\ell m}(\tilde{x}_{1}^{m})\prod_{j=1}^{n}W^{(\ell)}(y_{j}\mid\tilde{x}_{\tilde{s}_{j}})\right)}\\
 & =\sum_{x_{1}^{m}\in{\cal F}}\sum_{\tilde{x}_{1}^{m}\in({\cal X}^{\otimes\ell})^{\otimes m}}\sqrt{P_{X}^{\otimes\ell m}(x_{1}^{m})P_{X}^{\otimes\ell m}(\tilde{x}_{1}^{m})}\sum_{y_{1}^{n}\in({\cal Y}^{\otimes\ell})^{\otimes n}}\prod_{j=1}^{n}\sqrt{W^{(\ell)}(y_{j}\mid x_{s_{j}})W^{(\ell)}(y_{j}\mid\tilde{x}_{\tilde{s}_{j}})}\\
 & \trre[=,c]\sum_{x_{1}^{m}\in{\cal F}}\sum_{\tilde{x}_{1}^{m}\in({\cal X}^{\otimes\ell})^{\otimes m}}\sqrt{P_{X}^{\otimes\ell m}(x_{1}^{m})P_{X}^{\otimes\ell m}(\tilde{x}_{1}^{m})}\prod_{j=1}^{n}\left(\sum_{y\in{\cal Y}^{\otimes\ell}}\sqrt{W^{(\ell)}(y_{j}\mid x_{s_{j}})W^{(\ell)}(y_{j}\mid\tilde{x}_{\tilde{s}_{j}})}\right)\\
 & \trre[=,d]\sum_{x_{1}^{m}\in{\cal F}}\sum_{\tilde{x}_{1}^{m}\in({\cal X}^{\otimes\ell})^{\otimes m}}\sqrt{P_{X}^{\otimes\ell m}(x_{1}^{m})P_{X}^{\otimes\ell m}(\tilde{x}_{1}^{m})}\prod_{j=1}^{n}B(x_{s_{j}},\tilde{x}_{\tilde{s}_{j}})\\
 & \trre[=,e]\sum_{x_{1}^{m}\in{\cal F}}\sum_{\tilde{x}_{1}^{m}\in({\cal X}^{\otimes\ell})^{\otimes m}}\sqrt{P_{X}^{\otimes\ell m}(x_{1}^{m})P_{X}^{\otimes\ell m}(\tilde{x}_{1}^{m})}\prod_{i=1}^{m}\prod_{k=1}^{m}B^{t_{i\to k}}(x_{i},\tilde{x}_{k})\\
 & \trre[=,f]\sum_{x_{1}^{m}\in{\cal F}}\sqrt{P_{X}^{\otimes\ell m}(x_{1}^{m})}\underbrace{\sum_{\tilde{x}_{1}\in{\cal X}^{\otimes\ell}}\sqrt{P_{X}^{\otimes\ell}(\tilde{x}_{1})}\prod_{k=1}^{m}B^{t_{k\to1}}(x_{k},\tilde{x}_{1})}_{=f_{1}(x_{1}^{m})}\cdot\nonumber \\
 & \hphantom{==============}\underbrace{\sum_{\tilde{x}_{2}\in{\cal X}^{\otimes\ell}}\sqrt{P_{X}^{\otimes\ell}(\tilde{x}_{2})}\prod_{k=1}^{m}B^{t_{k\to2}}(x_{k},\tilde{x}_{2})}_{=f_{2}(x_{1}^{m})}\cdots\underbrace{\sum_{\tilde{x}_{m}\in{\cal X}^{\otimes\ell}}\sqrt{P_{X}^{\otimes\ell}(\tilde{x}_{m})}\prod_{k=1}^{m}B^{t_{k\to m}}(x_{k},\tilde{x}_{m})}_{=f_{m}(x_{1}^{m})}\\
 & \trre[=,g]\sum_{x_{1}^{m}\in{\cal F}}\sqrt{P_{X}^{\otimes\ell m}(x_{1}^{m})}\prod_{i=1}^{m}f_{i}(x_{1}^{m}),
\end{align}
where $(a)$ follows since in the denominator of the square-root term,
$P_{X}^{\otimes\ell m}(x_{1}^{m})\prod_{i=1}^{m}\prod_{j\in[n]\colon s_{j}=i}W^{(\ell)}(y_{j}\mid x_{i})>0$,
and so restricting to $x_{1}^{m}\in{\cal F}$ only increases the bound,
$(b)$ follows from $\sqrt{\sum z_{i}^{2}}\leq\sum z_{i}$, $(c)$
follows since for a separable function $f(z_{1},z_{2})=f(z_{1})f(z_{2})$
it holds that $\sum_{z_{1}}\sum_{z_{2}}f(z_{1})f(z_{2})=\sum_{z_{1}}f(z_{1})\sum_{z_{2}}f(z_{2})$,
$(d)$ follows from the definition of the Bhattacharyya coefficient
in (\ref{eq: Bhattacharyya coefficient}), $(e)$ follows by writing
the product over the $n$ reads according to the possible values of
$(s_{i},s_{k})$ and utilizing the definition of $\overline{\overline{t}}=(t_{i\to k}){}_{(i,k)\in[m]^{\otimes2}}$,
$(f)$ follows by re-ordering of the terms by grouping them according
to $\tilde{x}_{i}$, and $(g)$ follows by utilizing the definition
of $f_{i}(x_{1}^{m})$ in (\ref{eq: definition of F_i}), for $i\in[m]$.
\end{IEEEproof}
Our next step involves further upper bounding the $f_{i}(x_{1}^{m})$
terms, for any $i\in[m]$. Let us focus on a specific $f_{i}(x_{1}^{m})$.
In principle, we separate the summation involved into $m+1$ terms,
$\tilde{x}=x_{i}$ for $i\in[m]$, and the remaining $\tilde{x}\in{\cal X}^{\otimes\ell}\backslash\cup_{i\in[m]}\{x_{i}\}$.
For technical reasons, we inflate each of the first $m$ terms into
a small Hamming ball around each $x_{i}$. To this end, choose $\delta>0$
small enough (yet fixed with $\ell$), whose value will be discussed
in what follows. Let the Hamming ball centered at $x\in{\cal X}^{\otimes\ell}$
and radius $\delta\ell$ be
\[
{\cal B}_{\text{Ham}}(x,\delta):=\left\{ \tilde{x}\in{\cal X}^{\otimes\ell}\colon d_{\text{Ham}}(x,\tilde{x})\leq\delta\ell\right\} .
\]
Define the union of the Hamming balls centered at $x_{1}^{m}$ as
\[
{\cal U}(x_{1}^{m}):=\bigcup_{i\in[m]}{\cal B}_{\text{Ham}}(x_{i},\delta).
\]
Now, as mentioned, we write the summation defining $f_{i}(x_{1}^{m})$
as a sum of $m$ terms (over $m$ Hamming balls ${\cal B}_{\text{Ham}}(x_{j},\delta)$)
and then over the rest of the space ${\cal U}^{c}(x_{1}^{m})$, as
\begin{align}
f_{i}(x_{1}^{m}) & \leq\sum_{j=1}^{m}\sum_{\tilde{x}\in{\cal B}_{\text{Ham}}(x_{j},\delta)}\sqrt{P_{X}^{\otimes\ell}(\tilde{x})}\prod_{k=1}^{m}B^{t_{k\to i}}(x_{k},\tilde{x})+\sum_{\tilde{x}\in{\cal U}^{c}(x_{1}^{m})}\sqrt{P_{X}^{\otimes\ell}(\tilde{x})}\prod_{k=1}^{m}B^{t_{k\to i}}(x_{k},\tilde{x}),\label{eq: decomposition of F_i}
\end{align}
where the upper bound in (\ref{eq: decomposition of F_i}) is for
the case that the Hamming balls are not pairwise disjoint (though
for all $\delta$ small enough they will be pairwise disjoint). To
proceed, we begin with the first $m$ terms in (\ref{eq: decomposition of F_i}).
In Lemma \ref{lem: Bhat Hamming ball property}, we show that 
\[
B(x_{k},\tilde{x})\leq e^{\delta d_{B,\text{max}}\ell}\cdot B(x_{k},x_{j})
\]
for any $\tilde{x}\in{\cal B}_{\text{Ham}}(x_{j},\delta)$ (continuity
of the Bhattacharyya parameter in the Hamming distance), and $B(x_{k},\tilde{x})\leq1$
always holds. Thus, with a slight abuse of notation, we denote the
modified Bhattacharyya coefficient as 
\[
B_{\delta}(x,\tilde{x}):=\left(e^{\delta d_{B,\text{max}}\ell}\cdot B(x,\tilde{x})\right)\wedge1.
\]
Then, for any given $j\in[m]$, we bound
\begin{align}
 & \sum_{\tilde{x}\in{\cal B}_{\text{Ham}}(x_{j},\delta)}\sqrt{P_{X}^{\otimes\ell}(\tilde{x})}\prod_{k=1}^{m}B^{t_{k\to i}}(x_{k},\tilde{x})\nonumber \\
 & \trre[\leq,a]\sum_{\tilde{x}\in{\cal B}_{\text{Ham}}(x_{j},\delta)}\sqrt{P_{X}^{\otimes\ell}(\tilde{x})}\prod_{k=1}^{m}B_{\delta}^{t_{k\to i}}(x_{k},x_{j})\\
 & \trre[\leq,b]e^{\frac{1}{2}\delta\log(1/p_{\text{min}})\ell}\sum_{\tilde{x}\in{\cal B}_{\text{Ham}}(x_{j},\delta)}\sqrt{P_{X}^{\otimes\ell}(x_{j})}\prod_{k=1}^{m}B_{\delta}^{t_{k\to i}}(x_{k},x_{j})\\
 & =e^{\frac{1}{2}\delta\log(1/p_{\text{min}})\ell}\cdot\left|{\cal B}_{\text{Ham}}(x_{j},\delta)\right|\cdot\sqrt{P_{X}^{\otimes\ell}(x_{j})}\prod_{k=1}^{m}B_{\delta}^{t_{k\to i}}(x_{k},x_{j})\\
 & \trre[\leq,c]e^{\frac{1}{2}\delta\log(1/p_{\text{min}})\ell}\cdot e^{\ell h_{\text{bin}}(\delta)}\sqrt{P_{X}^{\otimes\ell}(x_{j})}\prod_{k=1}^{m}B_{\delta}^{t_{k\to i}}(x_{k},x_{j})\\
 & \trre[=,d]e^{g_{0}(\delta)\ell}\cdot\sqrt{P_{X}^{\otimes\ell}(x_{j})}\prod_{k=1}^{m}B_{\delta}^{t_{k\to i}}(x_{k},x_{j})\\
 & \trre[\leq,e]e^{g_{0}(\delta)\ell}\cdot\sqrt{P_{X}^{\otimes\ell}(x_{j})}\overline{B}_{\delta}^{\tilde{r}_{i}-t_{j\to i}}(x_{1}^{m}),\label{eq: upper bound on F Hammig ball terms}
\end{align}
where $(a)$ follows, as noted above, from Lemma \ref{lem: Bhat Hamming ball property},
$(b)$ follows similarly from Lemma \ref{lem: Input distribution continuity}
(continuity of the source sequence distribution in the Hamming distance),
where $p_{\text{min}}=\min_{x\in{\cal X}}P_{X}(x)>0$, $(c)$ follows
from the standard upper bound on the cardinality of the Hamming ball
$|{\cal B}_{\text{Ham}}(x,\delta)|\leq e^{\ell h_{\text{bin}}(\delta)}$,
in $(d)$ we have defined 
\[
g_{0}(\delta):=\frac{\delta}{2}\log(1/p_{\text{min}})+h_{\text{bin}}(\delta),
\]
and in $(e)$ we upper bound
\[
B_{\delta}(x_{k},x_{j})\leq\max_{k,j\in[m]\colon k\neq j}B_{\delta}(x_{k},x_{j}):=\overline{B}_{\delta}(x_{1}^{m}),
\]
for any $j,k\in[m]$ and $j\neq k$, and then further upper bound
\[
\prod_{k=1}^{m}B_{\delta}^{t_{k\to i}}(x_{k},x_{j})\leq\overline{B}_{\delta}^{\sum_{k\in[m]\colon k\neq j}t_{k\to i}}(x_{1}^{m})=\overline{B}_{\delta}^{\tilde{r}_{i}-t_{j\to i}}(x_{1}^{m}).
\]
Next, for the last sum in the decomposition (\ref{eq: decomposition of F_i}),
if $\tilde{x}\in{\cal U}^{c}(x_{1}^{m})$, then its Hamming distance
to each $\{x_{i}\}_{i\in[m]}$ is at least $\delta\ell$. We thus
bound, 
\begin{align}
 & \sum_{\tilde{x}\in{\cal U}^{c}(x_{1}^{m})}\sqrt{P_{X}^{\otimes\ell}(\tilde{x})}\prod_{k=1}^{m}B^{t_{k\to i}}(x_{k},\tilde{x})\nonumber \\
 & \trre[\leq,a]\sum_{\tilde{x}\in{\cal U}^{c}(x_{1}^{m})}\prod_{k=1}^{m}B^{t_{k\to i}}(x_{k},\tilde{x})\\
 & \trre[\leq,b]\sum_{\tilde{x}\in{\cal U}^{c}(x_{1}^{m})}\prod_{k=1}^{m}e^{-\delta d_{B,\text{min}}\ell\cdot t_{k\to i}}\\
 & =\sum_{\tilde{x}\in{\cal U}^{c}(x_{1}^{m})}e^{-\delta d_{B,\text{min}}\ell\cdot\sum_{k=1}^{m}t_{k\to i}}\\
 & =\left|{\cal U}^{c}(x_{1}^{m})\right|\cdot e^{-\delta d_{B,\text{min}}\ell\cdot\tilde{r}_{i}}\\
 & \trre[\leq,c]|{\cal X}|^{\ell}\cdot e^{-\delta d_{B,\text{min}}\ell\cdot\tilde{r}_{i}}\\
 & \trre[=,d]e^{-g_{1}(\tilde{r}_{i},\delta)\cdot\ell},\label{eq: upper bound on F residual term}
\end{align}
where $(a)$ follows by trivially upper bounding $P_{X}^{\otimes\ell}(\tilde{x})\leq1$,
$(b)$ follows from Lemma \ref{lem: large Hamming implies large Bhat },
$(c)$ follows since ${\cal U}^{c}(x_{1}^{m})\subset{\cal X}^{\otimes\ell}$,
and in $(d)$ we have defined 
\[
g_{1}(\tilde{r},\delta)=\delta d_{B,\text{min}}\tilde{r}-\log|{\cal X}|.
\]
Using (\ref{eq: upper bound on F Hammig ball terms}) and (\ref{eq: upper bound on F residual term})
in (\ref{eq: decomposition of F_i}), we obtain the upper bound 
\[
f_{i}(x_{1}^{m})\leq e^{g_{0}(\delta)\ell}\sum_{j=1}^{m}\sqrt{P_{X}^{\otimes\ell}(x_{j})}\overline{B}_{\delta}^{\tilde{r}_{i}-t_{j\to i}}(x_{1}^{m})+e^{-g_{1}(\tilde{r}_{i},\delta)\cdot\ell}.
\]
Furthermore, using the assumption that $\tilde{s}_{1}^{n}$ is such
that $\tilde{r}_{1}^{m}\in{\cal G}_{\eta}$ and so $\tilde{r}_{i}\geq\frac{n}{m}(1-\eta)$,
\[
g_{1}(\tilde{r}_{i},\delta)\geq\delta d_{B,\text{min}}\frac{n}{m}(1-\eta)-\log|{\cal X}|.
\]
For all sufficiently large $n>n_{0}(m,\delta,d_{B,\text{min}},\eta,|{\cal X}|)$
it holds that 
\[
g_{1}(\tilde{r}_{i},\delta)\ge\delta d_{B,\text{min}}\frac{n}{m}(1-2\eta):=g_{2}(\delta)\cdot n.
\]
Hence, we may continue to upper bound as
\begin{align}
f_{i}(x_{1}^{m}) & \leq e^{g_{0}(\delta)\ell}\sum_{j=1}^{m}\sqrt{P_{X}^{\otimes\ell}(x_{j})}\overline{B}_{\delta}^{\tilde{r}_{i}-t_{j\to i}}(x_{1}^{m})+e^{-g_{2}(\delta)\cdot n\ell}\\
 & \leq e^{g_{0}(\delta)\ell}m\cdot\sqrt{\max_{j\in[m]}P_{X}^{\otimes\ell}(x_{j})}\overline{B}_{\delta}^{\bigwedge_{j=1}^{m}\{\tilde{r}_{i}-t_{j\to i}\}}(x_{1}^{m})+e^{-g_{2}(\delta)\cdot n\ell}\label{eq: upper bound on individual F_i}
\end{align}
where the last inequality holds since $a^{q_{1}}+a^{q_{2}}\leq2\cdot a^{q_{1}\wedge q_{2}}$
for $a>0$ and $q_{1},q_{2}>0$. After obtaining the bound (\ref{eq: upper bound on individual F_i})
on each $f_{i}(x_{1}^{m})$ we return to the upper bound on the pairwise
error probability in Lemma \ref{lem: Bhat upper bound on the pairwise error probability},
which involves their product. We thus upper bound (\ref{eq: Bhat upper bound on pairwise error probability statement})
as
\begin{align}
 & \sum_{x_{1}^{m}\in{\cal F}}\sqrt{P_{X}^{\otimes\ell m}(x_{1}^{m})}\prod_{i=1}^{m}f_{i}(x_{1}^{m})\nonumber \\
 & \leq\sum_{x_{1}^{m}\in{\cal F}}\sqrt{P_{X}^{\otimes\ell m}(x_{1}^{m})}\prod_{i=1}^{m}\left[e^{g_{0}(\delta)\ell}m\sqrt{\max_{j\in[m]}P_{X}^{\otimes\ell}(x_{j})}\overline{B}_{\delta}^{\bigwedge_{j=1}^{m}\{\tilde{r}_{i}-t_{j\to i}\}}(x_{1}^{m})+e^{-g_{2}(\delta)\cdot n\ell}\right]\\
 & \leq e^{g_{3}(\delta)m\ell}\sum_{x_{1}^{m}\in{\cal F}}\sqrt{P_{X}^{\otimes\ell m}(x_{1}^{m})\left(\max_{j\in[m]}P_{X}^{\otimes\ell}(x_{j})\right)^{m}}\cdot\prod_{i=1}^{m}\left[\overline{B}_{\delta}^{\bigwedge_{j=1}^{m}\{\tilde{r}_{i}-t_{j\to i}\}}+e^{-g_{4}(\delta)\cdot n\ell}\right],\label{eq: upper bound on product of F_i with outer summation}
\end{align}
where 
\[
g_{3}(\delta)=g_{0}(\delta)+\frac{\log m}{\ell},
\]
and
\[
g_{2}(\delta)-\frac{m\log\left(\frac{1}{\max_{x\in{\cal X}}P_{X}(x)}\right)}{n}\geq g_{2}(\delta)-\frac{m\log|{\cal X}|}{n}=:g_{4}(\delta)
\]
The right-hand side of the bound in (\ref{eq: upper bound on product of F_i with outer summation})
includes the product $\prod_{i=1}^{m}$ of a sum of two terms, for
a total of $2^{m}$ terms. We will evaluate the contribution of this
upper bound to the upper bound in (\ref{eq:  main bound on the error probability fixed x1 and x2})
on the conditional error probability for two separate cases. First,
consider the case in which the second term $e^{-g_{4}(\delta)\cdot n\ell}$
is included in the product for the indices ${\cal M}_{2}\subset[m]$,
with $m_{2}=|{\cal M}_{2}|>0$, and let ${\cal M}_{1}:=[m]\backslash{\cal M}_{2}$
denote the rest of the indices. Then, the upper bound reads 
\begin{align}
 & e^{g_{3}(\delta)m\ell}\sum_{x_{1}^{m}\in{\cal F}}\sqrt{P_{X}^{\otimes\ell m}(x_{1}^{m})\left(\max_{j\in[m]}P_{X}^{\otimes\ell}(x_{j})\right)^{m}}\cdot\prod_{i\in{\cal M}_{1}}\left[\overline{B}_{\delta}^{\bigwedge_{j=1}^{m}\{\tilde{r}_{i}-t_{j\to i}\}}\right]e^{-m_{2}\cdot g_{4}(\delta)\cdot n\ell}\nonumber \\
 & \trre[\leq,a]e^{g_{3}(\delta)m\ell}\sum_{x_{1}^{m}\in{\cal F}}e^{-m_{2}\cdot g_{4}(\delta)\cdot n\ell}\\
 & \trre[\leq,c]|{\cal X}|^{m\ell}e^{g_{3}(\delta)m\ell}e^{-g_{4}(\delta)\cdot n\ell},
\end{align}
where $(a)$ follows by the trivial bounds $P_{X}^{\otimes\ell}(x)<1$
and $\overline{B}_{\delta}(x_{1}^{m})\leq1$, and $(b)$ follows since
$m_{2}\geq1$. There are $2^{m}-1$ such terms, and so the total contribution
of these terms to the upper bound (\ref{eq:  main bound on the error probability fixed x1 and x2})
is upper bounded as 
\begin{align}
 & 2^{m}\cdot\sum_{\overline{t}_{1\to}\in{\cal J}(r_{1})}\sum_{\overline{t}_{2\to}\in{\cal J}(r_{2})}\cdots\sum_{\overline{t}_{m\to}\in{\cal J}(r_{m})}\cdot\ensuremath{\binom{r_{1}}{\overline{t}_{1\to}}}\ensuremath{\binom{r_{2}}{\overline{t}_{2\to}}}\cdots\ensuremath{\binom{r_{m}}{\overline{t}_{m\to}}}\times|{\cal X}|^{m\ell}e^{g_{3}(\delta)m\ell}\cdot e^{-g_{4}(\delta)\cdot n\ell}\nonumber \\
 & \leq2^{m}\cdot m^{r_{1}}\cdot m^{r_{2}}\cdots m^{r_{m}}e^{\log|{\cal X}|\cdot m\ell}e^{g_{3}(\delta)m\ell}e^{-g_{4}(\delta)\cdot n\ell}\\
 & =2^{m}\cdot m^{n}e^{\log|{\cal X}|\cdot m\ell}e^{g_{3}(\delta)m\ell}\cdot\exp\left[-\delta d_{B,\text{min}}\frac{n}{m}(1-2\eta)\ell\right]\\
 & \leq\exp\left[-\delta d_{B,\text{min}}\frac{(1-2\eta)}{m}\cdot n\ell+m\log2+n\log m+(\log|{\cal X}|+g_{0}(\delta))m\ell\right]\\
 & =e^{g_{5}(\delta)m\ell}e^{-\delta d_{B,\text{min}}\frac{(1-2\eta)}{m}\cdot n\ell}\\
 & \leq e^{g_{5}(\delta)m\ell}e^{-\delta d_{B,\text{min}}\frac{n\ell}{2m}},\label{eq: small terms contribution to conditional error probability}
\end{align}
where 
\[
g_{5}(\delta):=\log|{\cal X}|+g_{0}(\delta)+\frac{\log2}{\ell}+\frac{n\log m}{m\ell},
\]
and $\eta\leq\frac{1}{2}$. We next consider the contribution of the
remaining term, for which ${\cal M}_{2}=\emptyset$. The corresponding
term in (\ref{eq: upper bound on product of F_i with outer summation})
is given by 
\begin{equation}
e^{g_{3}(\delta)m\ell}\sum_{x_{1}^{m}\in{\cal F}}\sqrt{P_{X}^{\otimes\ell m}(x_{1}^{m})\left(\max_{j\in[m]}P_{X}^{\otimes\ell}(x_{j})\right)^{m}}\cdot\overline{B}_{\delta}^{\sum_{i=1}^{m}\bigwedge_{j=1}^{m}\{\tilde{r}_{i}-t_{j\to i}\}}.\label{eq: upper bound on the product of F_i terms -- main term}
\end{equation}
Now, the power of the Bhattacharyya coefficient in this term is given
by 
\begin{align}
\sum_{i=1}^{m}\bigwedge_{j=1}^{m}\{\tilde{r}_{i}-t_{j\to i}\} & =\min_{\sigma:[m]\to[m]}\sum_{i=1}^{m}\left\{ \tilde{r}_{i}-t_{\sigma(i)\to i}\right\} \\
 & =\min_{\sigma:[m]\to[m]}\sum_{j=1}^{m}\left(r_{j}-\sum_{i:\sigma(i)=j}t_{j\to i}\right).
\end{align}
Using (\ref{eq: upper bound on the product of F_i terms -- main term}),
the total contribution of this term to the upper bound (\ref{eq:  main bound on the error probability fixed x1 and x2})
is given by 
\begin{multline}
e^{g_{3}(\delta)m\ell}\sum_{x_{1}^{m}\in{\cal F}}\sqrt{P_{X}^{\otimes\ell m}(x_{1}^{m})\left(\max_{j\in[m]}P_{X}^{\otimes\ell}(x_{j})\right)^{m}}\\
\times\left(\sum_{\overline{t}_{1\to}\in{\cal J}(r_{1})}\sum_{\overline{t}_{2\to}\in{\cal J}(r_{2})}\cdots\sum_{\overline{t}_{m\to}\in{\cal J}(r_{m})}\cdot\ensuremath{\binom{r_{1}}{\overline{t}_{1\to}}}\ensuremath{\binom{r_{2}}{\overline{t}_{2\to}}}\cdots\ensuremath{\binom{r_{m}}{\overline{t}_{m\to}}}\max_{\sigma:[m]\to[m]}\overline{B}_{\delta}^{\sum_{j=1}^{m}\left(r_{j}-\sum_{i:\sigma(i)=j}t_{j\to i}\right)}-1\right)\label{conditional error probability  -- main term contribution}
\end{multline}
where the subtraction of $1$ corresponds to the case $t_{i\to i}=r_{i}$
for all $i\in[m]$, for which the power of $\overline{B}_{\delta}$
is $0$, yet it does not constitute an error (which is a result of
the indicator $\I\{\overline{\overline{t}}(s_{1}^{n},\tilde{s}_{1}^{n})\not\in{\cal T}_{0}\}$
in the upper bound (\ref{eq:  main bound on the error probability fixed x1 and x2})).
Continuing to upper bound the maximum over $\sigma:[m]\to[m]$ with
a summation $\sum_{\sigma:[m]\to[m]}$ and interchanging summation
order, we further upper bound (\ref{conditional error probability  -- main term contribution})
as 
\begin{align}
 & e^{g_{3}(\delta)m\ell}\sum_{\sigma:[m]\to[m]}\sum_{x_{1}^{m}\in{\cal F}}\sqrt{P_{X}^{\otimes\ell m}(x_{1}^{m})\left(\max_{j\in[m]}P_{X}^{\otimes\ell}(x_{j})\right)^{m}}\nonumber \\
 & \times\left(\sum_{\overline{t}_{1\to}\in{\cal J}(r_{1})}\ensuremath{\binom{r_{1}}{\overline{t}_{1\to}}}\overline{B}_{\delta}^{\left(r_{1}-\sum_{i:\sigma(i)=1}t_{1\to i}\right)}\sum_{\overline{t}_{2\to}\in{\cal J}(r_{2})}\binom{r_{2}}{\overline{t}_{2\to}}\overline{B}_{\delta}^{\left(r_{2}-\sum_{i:\sigma(i)=2}t_{2\to i}\right)}\cdots\right.\nonumber \\
 & \hphantom{================================}\left.\sum_{\overline{t}_{m\to}\in{\cal J}(r_{m})}\ensuremath{\binom{r_{m}}{\overline{t}_{m\to}}}\overline{B}_{\delta}^{\left(r_{m}-\sum_{i:\sigma(i)=m}t_{m\to i}\right)}-1\right)\\
 & =e^{g_{3}(\delta)m\ell}\sum_{\sigma:[m]\to[m]}\sum_{x_{1}^{m}\in{\cal F}}\sqrt{P_{X}^{\otimes\ell m}(x_{1}^{m})\left(\max_{j\in[m]}P_{X}^{\otimes\ell}(x_{j})\right)^{m}}\left(\prod_{j=1}^{m}\sum_{\overline{t}_{j\to}\in{\cal J}(r_{j})}\ensuremath{\binom{r_{j}}{\overline{t}_{j\to}}}\left(\overline{B}_{\delta}(x_{1}^{m})\right)^{r_{j}-\sum_{i:\sigma(i)=j}t_{j\to i}}-1\right).\label{conditional error probability  -- main term contribution continue 1}
\end{align}
We now split the summation over $\sigma:[m]\to[m]$ into the case
that $\sigma$ is a permutation, and when it is not. First, assume
that $\sigma\in\Pi_{m}$, that is, $\sigma$ is a permutation. So,
the power of $\overline{B}_{\delta}(x_{1}^{m})$ in (\ref{conditional error probability  -- main term contribution continue 1})
is given by $r_{j}-t_{\sigma^{-1}(i)\to i}$, in which only a single
term is subtracted from $r_{j}$. Then, the inner term in (\ref{conditional error probability  -- main term contribution continue 1})
is upper bounded as 
\begin{align}
 & \prod_{j=1}^{m}\sum_{\overline{t}_{j\to}\in{\cal J}(r_{j})}\ensuremath{\binom{r_{j}}{\overline{t}_{j\to}}}\left(\overline{B}_{\delta}(x_{1}^{m})\right)^{r_{j}-t_{\sigma^{-1}(i)\to i}}\nonumber \\
 & \trre[=,a]\prod_{j=1}^{m}\overline{B}_{\delta}^{r_{j}}(x_{1}^{m})\cdot\left(\overline{B}_{\delta}^{-1}(x_{1}^{m})+m-1\right)^{r_{j}}\\
 & =\prod_{j=1}^{m}\left(1+(m-1)\overline{B}_{\delta}(x_{1}^{m})\right)^{r_{j}}\\
 & \trre[=,b]\left[1+(m-1)\overline{B}_{\delta}(x_{1}^{m})\right]^{n}\\
 & \le\left[1+m\overline{B}_{\delta}(x_{1}^{m})\right]^{n}
\end{align}
where $(a)$ follows from the multinomial theorem, $(b)$ follows
since $\sum r_{j}=n$. Substituting this back to (\ref{conditional error probability  -- main term contribution continue 1}),
we obtain 
\begin{align}
 & e^{g_{3}(\delta)m\ell}\sum_{\sigma\in\Pi_{m}}\sum_{x_{1}^{m}\in{\cal F}}\sqrt{P_{X}^{\otimes\ell m}(x_{1}^{m})\left(\max_{j\in[m]}P_{X}^{\otimes\ell}(x_{j})\right)^{m}}\left[\left(1+m\overline{B}_{\delta}(x_{1}^{m})\right)^{n}-1\right]\nonumber \\
 & \leq e^{g_{3}(\delta)m\ell}m^{m}\sum_{x_{1}^{m}\in{\cal F}}\sqrt{P_{X}^{\otimes\ell m}(x_{1}^{m})\left(\max_{j\in[m]}P_{X}^{\otimes\ell}(x_{j})\right)^{m}}\left[\left(1+m\overline{B}_{\delta}(x_{1}^{m})\right)^{n}-1\right]\\
 & \trre[\leq,a]e^{g_{3}(\delta)m\ell}m^{m}\sum_{x_{1}^{m}\in{\cal F}}\sum_{j=1}^{m}\sqrt{P_{X}^{\otimes\ell m}(x_{1}^{m})\left(P_{X}^{\otimes\ell}(x_{j})\right)^{m}}\left[\left(1+m\overline{B}_{\delta}(x_{1}^{m})\right)^{n}-1\right]\\
 & \trre[=,b]e^{g_{3}(\delta)m\ell}m^{m+1}\sum_{x_{1}^{m}\in{\cal F}}\sqrt{P_{X}^{\otimes\ell m}(x_{1}^{m})\left(P_{X}^{\otimes\ell}(x_{1})\right)^{m}}\left[\left(1+m\overline{B}_{\delta}(x_{1}^{m})\right)^{n}-1\right]\\
 & =e^{g_{3}(\delta)m\ell}m^{m+1}\E_{X_{1}^{m}\sim P_{X}^{\otimes\ell m}}\left[\sqrt{\frac{\left(P_{X}^{\otimes\ell}(X_{1})\right)^{m}}{P_{X}^{\otimes\ell m}(X_{1}^{m})}}\left[\left(1+m\overline{B}_{\delta}(x_{1}^{m})\right)^{n}-1\right]\cdot\I\{X_{1}^{m}\in{\cal F}\}\right]\\
 & \trre[=,c]e^{g_{6}(\delta)m\ell}\E_{X_{1}^{m}\sim P_{X}^{\otimes\ell m}}\left[\sqrt{\frac{\left(P_{X}^{\otimes\ell}(X_{1})\right)^{m}}{P_{X}^{\otimes\ell m}(X_{1}^{m})}}\left[\left(1+m\overline{B}_{\delta}(x_{1}^{m})\right)^{n}-1\right]\cdot\I\{X_{1}^{m}\in{\cal F}\}\right],\label{eq: main term contribution 2}
\end{align}
where $(a)$ follows by upper bounding the maximum over $j\in[m]$
with a summation, $(b)$ follows from symmetry of the summand as a
function of $x_{1}^{m}$ to permutations (specifically, to the specific
index $j$), and $(c)$ follows by defining
\[
g_{6}(\delta)=g_{3}(\delta)+\frac{\frac{m+1}{m}\cdot\log m}{\ell}.
\]
Second, suppose that $\sigma\not\in\Pi_{m}$, that is $\sigma$ is
not a permutation. Then, there exists $j^{*}\in[m]$ such that $\{i:\sigma(i)=j^{*}\}=\emptyset$,
and so the power of $\overline{B}_{\delta}(x_{1}^{m})$ in (\ref{conditional error probability  -- main term contribution continue 1})
is given by $r_{j^{*}}$. Then, the inner term in (\ref{conditional error probability  -- main term contribution continue 1})
is upper bounded as 
\begin{align}
 & \prod_{j=1}^{m}\sum_{\overline{t}_{j\to}\in{\cal J}(r_{j})}\ensuremath{\binom{r_{j}}{\overline{t}_{j\to}}}\left(\overline{B}_{\delta}(x_{1}^{m})\right)^{r_{j}-t_{\sigma^{-1}(i)\to i}}\nonumber \\
 & \trre[=,a]\prod_{j\in[m]\colon j\neq j^{*}}^{m}\sum_{\overline{t}_{j\to}\in{\cal J}(r_{j})}\ensuremath{\binom{r_{j}}{\overline{t}_{j\to}}}\left(\overline{B}_{\delta}(x_{1}^{m})\right)^{r_{j}-t_{\sigma^{-1}(i)\to i}}\times\sum_{\overline{t}_{j^{*}\to}\in{\cal J}(r_{j^{*}})}\ensuremath{\binom{r_{j^{*}}}{\overline{t}_{j^{*}\to}}}\left(\overline{B}_{\delta}(x_{1}^{m})\right)^{r_{j^{*}}}\\
 & \trre[\leq,b]\prod_{j\in[m]\colon j\neq j^{*}}^{m}\sum_{\overline{t}_{j\to}\in{\cal J}(r_{j})}\ensuremath{\binom{r_{j}}{\overline{t}_{j\to}}}\times\sum_{\overline{t}_{j^{*}\to}\in{\cal J}(r_{j^{*}})}\ensuremath{\binom{r_{j^{*}}}{\overline{t}_{j^{*}\to}}}\left(\overline{B}_{\delta}(x_{1}^{m})\right)^{r_{j^{*}}}\\
 & \trre[=,c]\prod_{j\in[m]\colon j\neq j^{*}}^{m}m^{r_{j}}\times m^{r_{j^{*}}}\cdot\left(\overline{B}_{\delta}(x_{1}^{m})\right)^{r_{j^{*}}}\\
 & \trre[\leq,d]m^{n}\cdot\left(\overline{B}_{\delta}(x_{1}^{m})\right)^{r_{j^{*}}}\\
 & \trre[\leq,e]m^{n}\cdot\left(\overline{B}_{\delta}(x_{1}^{m})\right)^{(1-\eta)\frac{n}{m}}\\
 & \trre[\leq,f]m^{n}\cdot\left(\overline{B}_{\delta}(x_{1}^{m})\right)^{\frac{n}{2m}},
\end{align}
where $(a)$ follows by separating the product to $j\neq j^{*}$ ($m-1$
terms) and $j=j^{*}$, $(b)$ follows since $\overline{B}_{\delta}(x_{1}^{m})\leq1$,
$(c)$ from the multinomial theorem, $(d)$ since $\sum r_{j}\leq n$,
$(e)$ follows since $r_{1}^{m}\in{\cal G}_{\eta}$ by assumption,
and so $r_{j^{*}}\geq(1-\eta)\frac{n}{m}$, and $(f)$ follows since
$\eta\leq\frac{1}{2}$. 

Substituting this upper bound back to (\ref{conditional error probability  -- main term contribution continue 1}),
we now obtain that it is upper bounded as 
\begin{align}
 & e^{g_{3}(\delta)m\ell}\sum_{\sigma:[m]\to[m]\not\in\Pi_{m}}\sum_{x_{1}^{m}\in{\cal F}}\sqrt{P_{X}^{\otimes\ell m}(x_{1}^{m})\left(\max_{j\in[m]}P_{X}^{\otimes\ell}(x_{j})\right)^{m}}m^{n}\cdot\left(\overline{B}_{\delta}(x_{1}^{m})\right)^{\frac{n}{2m}}\nonumber \\
 & \trre[\leq,a]e^{g_{3}(\delta)m\ell}\sum_{\sigma:[m]\to[m]\not\in\Pi_{m}}\sum_{x_{1}^{m}\in{\cal F}}m^{n}\cdot\left(\overline{B}_{\delta}(x_{1}^{m})\right)^{\frac{n}{2m}}\\
 & \trre[\leq,b]e^{g_{3}(\delta)m\ell}\cdot m^{m}\cdot|{\cal X}|^{\ell m}\left(\overline{B}_{\delta}(x_{1}^{m})\right)^{\frac{n}{2m}}\\
 & =e^{g_{7}(\delta)m\ell}\left(\overline{B}_{\delta}(x_{1}^{m})\right)^{\frac{n}{2m}},\label{eq: main term contribution 3}
\end{align}
where $(a)$ follows since $P_{X}^{\otimes\ell}(x)\leq1$, $(b)$
since ${\cal F}\subseteq{\cal X}^{\otimes\ell m}$, and $(c)$ follows
by defining 
\[
g_{7}(\delta)=g_{3}(\delta)+\frac{\log m}{\ell}+\log|{\cal X}|.
\]
Adding (\ref{eq: small terms contribution to conditional error probability}),
(\ref{eq: main term contribution 2}) and \ref{eq: main term contribution 3},
we obtain from (\ref{eq:  main bound on the error probability fixed x1 and x2})
the bound
\begin{multline}
\P[{\cal E}\cap{\cal F}\mid S_{1}^{n}=s_{1}^{n}]\leq e^{g_{5}(\delta)m\ell}e^{-\delta d_{B,\text{min}}\frac{n\ell}{2m}}\\
+e^{g_{6}(\delta)m\ell}\E_{X_{1}^{m}\sim P_{X}^{\otimes\ell m}}\left[\sqrt{\frac{\left(P_{X}^{\otimes\ell}(X_{1})\right)^{m}}{P_{X}^{\otimes\ell m}(X_{1}^{m})}}\left[\left(1+m\overline{B}_{\delta}(X_{1}^{m})\right)^{n}-1\right]\cdot\I\{X_{1}^{m}\in{\cal F}\}\right]\\
+e^{g_{7}(\delta)m\ell}\left(\overline{B}_{\delta}(x_{1}^{m})\right)^{\frac{n}{2m}}.\label{eq: bound on pairwise error probability penultimate}
\end{multline}
We now recall the definitions of $g_{0},g_{1},\ldots,g_{7}$. It holds
that 
\[
g_{6}(\delta)=\frac{\delta}{2}\log(1/p_{\text{min}})+h_{\text{bin}}(\delta)+\frac{\log m}{\ell}+\frac{\frac{m+1}{m}\cdot\log m}{\ell},
\]
and since $\ell=\beta\log n$, for all $n$ sufficiently large it
holds that 
\[
g_{6}(\delta)\leq\delta\log(1/p_{\text{min}})+h_{\text{bin}}(\delta)=:c_{0}(\delta).
\]
We also recall that $B_{\delta}(x,\tilde{x}):=(e^{\delta d_{B,\text{max}}\ell}\cdot B(x,\tilde{x}))\wedge1$
and note that 
\[
m\overline{B}_{\delta}(x_{1}^{m})\leq e^{\left(\frac{\log m}{\ell}+\delta d_{B,\text{max}}\right)\ell}\cdot\overline{B}(x_{1}^{m})\leq e^{c_{1}(\delta)\ell}\overline{B}(x_{1}^{m}),
\]
where $c_{1}(\delta):=2\delta d_{B,\text{max}}$, assuming $n$ sufficiently
large. Similarly, 
\[
g_{5}(\delta):=\log|{\cal X}|+\frac{\delta}{2}\log(1/p_{\text{min}})+h_{\text{bin}}(\delta)+\frac{\log2}{\ell}+\frac{n\log m}{m\ell},
\]
and so 
\[
\delta d_{B,\text{min}}\frac{n\ell}{2m}-g_{5}(\delta)m\ell\ge\delta d_{B,\text{min}}\frac{n\ell}{4m}:=c_{2}(\delta)n\ell
\]
for all $n$ sufficiently large. Finally, 
\begin{align}
e^{g_{7}(\delta)m\ell}\left(\overline{B}_{\delta}(x_{1}^{m})\right)^{\frac{n}{2m}} & \leq e^{g_{7}(\delta)m\ell}e^{\delta d_{B,\text{max}}\frac{n\ell}{m}}\left(\overline{B}(x_{1}^{m})\right)^{\frac{n}{2m}}\\
 & \le e^{c_{1}(\delta)n\ell}\left(\overline{B}(x_{1}^{m})\right)^{\frac{n}{2m}},
\end{align}
using 
\begin{align}
 & \frac{1}{n\ell}\left(g_{7}(\delta)m\ell+\delta d_{B,\text{max}}\frac{n\ell}{m}\right)\nonumber \\
 & \leq\frac{g_{7}(\delta)m}{n}+\frac{\delta d_{B,\text{max}}}{m}\\
 & \leq\frac{2\delta d_{B,\text{max}}}{m}\\
 & \leq c_{1}(\delta),
\end{align}
where the inequality holds for all sufficiently large $n$. Thus,
(\ref{eq: bound on pairwise error probability penultimate}) is further
simplified to 
\begin{align}
\P[{\cal E}\cap{\cal F}\mid S_{1}^{n}=s_{1}^{n}] & \leq e^{c_{0}(\delta)m\ell}\E_{X_{1}^{m}\sim P_{X}^{\otimes\ell m}}\left[\sqrt{\frac{\left(P_{X}^{\otimes\ell}(X_{1})\right)^{m}}{P_{X}^{\otimes\ell m}(X_{1}^{m})}}\left[\left(1+e^{c_{1}(\delta)\ell}\cdot\overline{B}(X_{1}^{m})\right)^{n}-1\right]\cdot\I\{X_{1}^{m}\in{\cal F}\}\right]\nonumber \\
 & \hphantom{==}+e^{c_{1}(\delta)n\ell}\left(\overline{B}(x_{1}^{m})\right)^{\frac{n}{2m}}+e^{-c_{2}(\delta)n\ell}.
\end{align}
Using (\ref{eq: main decomposition of error probability}) completes
the proof of Proposition \ref{prop: Raw upper bound on the error probability}. 

\section{Conclusion and Future Research \label{sec:conclusion}}

We derived upper and lower bounds on the error probability for clustering
$n$ noisy short reads into $m$ unknown source sequences, observed
through a memoryless channel. While both bounds decay exponentially
with the sequence length $\ell$, and the lower and upper bounds on
the exponent are given by single-letter expressions, the bounds do
not necessarily match. This gap is also reflected in the proof, which
is technically challenging and requires various compromises along
the derivation. It is therefore of interest to find alternative analysis
techniques that reduce, or possibly close, the gap between the upper
and lower bounds. A natural direction for future research is the analysis
of the regime in which $m$ scales with $n$. In our proof, the main
issue is that various steps include factors as large as $e^{m\log m}$.
Thus, the proof directly extends to the case in which $\log m=o(\ell)=o(\log n)$,
that is, $m$ increases sub-polynomially with $n$. It is therefore
of interest to extend the analysis to faster scaling, up to $m=\Theta(n)$.
Nonetheless, when the number of source sequences is large, it is conceivable
that \emph{perfect} clustering is not likely, and so one may need
to settle for approximate clustering, which allows for a restricted
number of erroneous assignments. It is also of interest to explore
clustering under different reading channel models, e.g., channels
with synchronization errors (insertions, deletions). Finally, it is
also of interest to analyze the error probability of practical clustering
algorithms, and compare them to the performance of the optimal clustering
rule.

\appendices{\numberwithin{equation}{section}}

\section{Proof of Theorem \ref{thm: Exponent for a uniform source} \label{sec: proof of theorem exponent lower bound}}

We evaluate the asymptotics of the bound in (\ref{eq: upper bound on clustering error finite ell})
of Proposition \ref{prop: Raw upper bound on the error probability},
for the choice of ${\cal F}={\cal D}_{\tau}\cap{\cal R}_{\rho}$,
using the definitions in (\ref{eq: high probability set for a uniform source})
and (\ref{eq: high probability source distributions}). We first show
that the last three terms in (\ref{eq: upper bound on clustering error finite ell})
decay super-exponentially with $\ell$, and are thus negligible. To
this end, let $\delta>0$ be small enough so that 
\[
c_{1}(\delta)=2\delta d_{B,\text{\emph{max}}}\leq\frac{\tau}{4m}.
\]
Then, the term $e^{-c_{2}(\delta)n\ell}$, the term
\[
e^{c_{1}(\delta)n\ell}\left(\overline{B}(x_{1}^{m})\right)^{\frac{n}{2m}}\leq e^{c_{1}(\delta)n\ell}e^{-\frac{\tau}{2m}n\ell},
\]
and $me^{-c_{3}\eta^{2}n}$ all decay super-exponentially with $\ell$.
Next, (\ref{eq: large deviations for set F uniform}) shows that $\P[X_{1}^{m}\in{\cal D}_{\tau}^{c}]\leq e^{-E_{\text{B}}(\tau)\cdot\ell+o(\ell)}$,
that is, decays exponentially in $\ell$ with exponent $E_{\text{B}}(\tau)$.
Similarly, $\P[X_{1}^{m}\in{\cal R}_{\rho}^{c}]\leq e^{-E_{\text{P}}(\rho)\cdot\ell+o(\ell)}$.
Thus, 
\[
\P[X_{1}^{m}\in{\cal F}^{c}]\leq e^{-[E_{\text{B}}(\tau)\wedge E_{\text{P}}(\rho)]\cdot\ell+o(\ell)}.
\]
It thus remains to analyze the expectation term in (\ref{eq: upper bound on clustering error finite ell}),
which, in the case of a uniform source sequence distribution, is given
by 
\begin{align}
 & e^{c_{0}(\delta)m\ell}\E_{X_{1}^{m}\sim P_{X}^{\otimes\ell m}}\left[\sqrt{\frac{\left(P_{X}^{\otimes\ell}(X_{1})\right)^{m-1}}{\prod_{i=2}^{m}P_{X}^{\otimes\ell}(X_{i})}}\left[\left(1+e^{c_{1}(\delta)\ell}\overline{B}(X_{1}^{m})\right)^{n}-1\right]\cdot\I\{X_{1}^{m}\in{\cal D}_{\tau}\}\I\{X_{1}^{m}\in{\cal R}_{\rho}\}\right]\nonumber \\
 & \trre[\leq,*]e^{c_{0}(\delta)m\ell}\E_{X_{1}^{m}\sim P_{X}^{\otimes\ell m}}\left[e^{\rho\ell}\left[\left(1+e^{c_{1}(\delta)\ell}e^{-\tau\ell}\right)^{n}-1\right]\cdot\right]\\
 & \leq e^{c_{0}(\delta)m\ell}e^{\rho\ell}\left[\left(1+e^{c_{1}(\delta)\ell}e^{-\tau\ell}\right)^{n}-1\right]
\end{align}
where $(*)$ follows since $\overline{B}(X_{1}^{m})\leq e^{-\tau\ell}$
on ${\cal D}_{\tau}$. Now, for any given $\tau-\rho>1/\beta$, there
exists $\delta>0$ sufficiently small so that the exponent of the
term in the parenthesis is strictly positive, that is (with a slight
abuse of notation) 
\[
\tau(\delta):=\tau-c_{1}(\delta)>\frac{1}{\beta}.
\]
Lemma \ref{lem: binomial asymptotics} then shows that 
\begin{align}
 & \left(1+e^{c_{1}(\delta)\ell}e^{-\tau\ell}\right)^{n}-1\nonumber \\
 & =\left(1+e^{-\tau(\delta)\cdot\ell}\right)^{n}-1\\
 & \trre[=,a]\left(1+\frac{1}{n^{\tau(\delta)\cdot\beta}}\right)^{n}-1\\
 & \trre[=,b]\frac{1}{n^{\tau(\delta)\cdot\beta-1}}+O\left(\frac{1}{n^{2(\tau(\delta)\cdot\beta-1)}}\right)\\
 & \trre[=,c]e^{-(\tau(\delta)-1/\beta)\ell}+e^{-2(\tau(\delta)-1/\beta)\ell+o(\ell)},
\end{align}
where $(a)$ follows from $\ell=\beta\log n$, $(b)$ follows from
the asymptotic expansion derived in Lemma \ref{lem: binomial asymptotics},
which holds since $\tau(\delta)\cdot\beta>1$, and $(c)$ holds using
$n=e^{\ell/\beta}$. Thus, for the chosen $\delta,\eta$, it holds
that the upper bound (\ref{eq: upper bound on clustering error finite ell})
of Proposition \ref{prop: Raw upper bound on the error probability}
is 
\[
p_{\text{error}}(\mathsf{C}_{\text{MAP}})\leq e^{c_{0}(\delta)m\ell}\left(e^{-(\tau(\delta)-\rho-1/\beta)\ell}+e^{-2(\tau(\delta)-1/\beta)\ell+o(\ell)}\right)+e^{-[E_{\text{B}}(\tau)\wedge E_{\text{P}}(\rho)]\cdot\ell+o(\ell)}+e^{-\Omega(n\ell)}.
\]
Hence, 
\[
\lim_{\ell\to\infty}-\frac{1}{\ell}\log p_{\text{error}}(\mathsf{C}_{\text{MAP}})\geq\left[\tau(\delta)-\rho-\frac{1}{\beta}\right]\wedge E_{\text{B}}(\tau)\wedge E_{\text{P}}(\rho).
\]
Since $\eta,\delta>0$ can be made arbitrarily small, and $\lim_{\eta\downarrow0}\lim_{\delta\downarrow0}\tau(\eta,\delta)=\tau$,
an achievable exponent is thus 
\[
\left[\tau-\rho-\frac{1}{\beta}\right]\wedge E_{\text{B}}(\tau)\wedge E_{\text{P}}(\rho).
\]
This proves (\ref{eq: exponent lower bound general}). Finally, in
the case of a uniform source, it holds that $E_{\text{P}}(\rho)$
is unbounded for any $\rho>0$. Taking an arbitrarily small $\rho>0$,
an achievable exponent is thus 
\[
\left[\tau-\frac{1}{\beta}\right]\wedge E_{\text{B}}(\tau).
\]
Since $\tau\to\tau-\frac{1}{\beta}$ is monotonically increasing,
and $\tau\to E_{\text{B}}(\tau)$ is monotonically decreasing (see
(\ref{eq: exponent of Bhatacharrya})), the maximal exponent is achieved
for $\tau^{*}$ as the solution to $\tau-\frac{1}{\beta}=E_{\text{B}}(\tau)$,
and this proves (\ref{eq: exponent lower bound uniform}). 

\section{Proof of Theorem \ref{thm: upper bound on the exponent} \label{sec: proof of theorem exponent upper bound}}

We begin with the proof of Proposition \ref{prop: asignment versus clustering error probability}. 
\begin{IEEEproof}[Proof of Proposition \ref{prop: asignment versus clustering error probability}]
We construct an assignment rule $\mathsf{A}^{\sharp}$, which is
possibly suboptimal, and so, trivially, $\tilde{p}_{\text{error}}(\mathsf{A}_{\text{MAP}})\leq\tilde{p}_{\text{error}}(\mathsf{A}^{\sharp})$.
To begin, if the sequences $\{X_{i}\}_{i\in[m]}$ are not all distinct,
then the assignment rule $\mathsf{A}^{\sharp}$ declares an error.
In all other cases, the main idea is to utilize the clustering rule
$\mathsf{C}_{\text{MAP}}$ to construct $\mathsf{A}^{\sharp}$, as
follows: First, the clustering rule operates on $Y_{1}^{n}$, while
ignoring $X_{1}^{m}$, and we obtain $\tilde{S}_{1}^{n}=\mathsf{C}_{\text{MAP}}(Y_{1}^{n})$.
Second, for $i\in[m]$, let 
\[
{\cal \tilde{J}}_{i}:=\left\{ j\in[n]\colon\tilde{S}_{j}=i\right\} 
\]
denote the indices in $[n]$ that belong to the $i$th cluster. Then,
for each $i\in[m]$, we detect the source sequence from the reads
that belong to the $i$th cluster, by some detection rule $\mathsf{D}\colon\bigcup_{k=1}^{\infty}({\cal Y}^{\otimes\ell})^{\otimes k}\to{\cal X}^{\otimes\ell}$
\[
\hat{X}_{i}=\mathsf{D}(\{Y_{j}\}_{j\in{\cal \tilde{J}}_{i}}).
\]
Third, if $\{\hat{X}_{i}\}_{i\in[m]}=\{X_{i}\}_{i\in[m]}$, that is,
the \emph{set} of estimated sequences matches the set of source sequences,
then we find the permutation $\pi\in\Pi_{m}$ such that 
\[
\hat{X}_{\pi(i)}=X_{i}
\]
for all $i\in[m]$, and output $\mathsf{A}^{\sharp}(Y_{1}^{n})=\hat{S}_{1}^{n}=\pi(\tilde{S}_{1}^{n})=(\pi(\tilde{S}_{1}),\pi(\tilde{S}_{2}),\ldots\pi(\tilde{S}_{n}))$
as the decoded assignment. Otherwise, we declare an error. 

We next upper bound the error probability of $\mathsf{A}^{\sharp}$.
Let 
\[
{\cal D}:=\left\{ \bigcup_{i,j\in[m]\colon i<j}\{X_{i}=X_{j}\}\right\} 
\]
be the event in which some pair of source sequences is identical,
for which it holds that 
\[
\P[{\cal D}]\leq\frac{m(m-1)}{2}e^{-\ell H_{2}(P_{X})}\leq m^{2}e^{-\ell H_{2}(P_{X})},
\]
from the union bound, and the definition and tensorization of the
second-order R\'{e}nyi entropy (also known as the collision entropy)
(\ref{eq: Second order Renyi entropy}). Let $\eta\in(0,1)$ be given
and let us recall the set from the proof of Proposition \ref{prop: Raw upper bound on the error probability},
\begin{equation}
{\cal G}_{\eta}:=\left\{ r_{1}^{m}\in\mathbb{N}_{+}^{\otimes m}\colon\left\{ \left|r_{i}-\frac{n}{m}\right|\leq\eta\cdot\frac{n}{m}\right\} \text{for all }i\in[m]\right\} ,\label{eq: concentration of sampling histogram event lower bound}
\end{equation}
for which it holds that $\P\left[R_{1}^{m}\not\in{\cal G}_{\eta}\right]\leq me^{-c\eta^{2}n}$.
With a slight abuse of notation, we will consider this an event ${\cal G}_{\eta}$.
Further, let 
\[
{\cal E}_{A}:=\left\{ \mathsf{A}^{\sharp}(X_{1}^{m},Y_{1}^{n})\neq S_{1}^{n}\right\} 
\]
denote an assignment error event for $\mathsf{A}^{\sharp}$, and let
\[
{\cal E}_{C}:=\left\{ \bigcap_{\pi\in\Pi_{m}}\left\{ \mathsf{C}_{\text{MAP}}(Y_{1}^{n})\neq\pi(S_{1}^{n})\right\} \right\} 
\]
denote a clustering error event for $\mathsf{C}_{\text{MAP}}$. Also,
let $\pi^{*}\in\Pi_{m}$ denote the permutation that satisfies 
\[
S_{1}^{n}=\pi^{*}(\tilde{S}_{1}^{n})=(\pi^{*}(\tilde{S}_{1}),\pi^{*}(\tilde{S}_{2}),\ldots\pi^{*}(\tilde{S}_{n})),
\]
if such a permutation exists, and otherwise any arbitrary permutation.
So, whenever the clustering is correct, the assignment is also correct
if $\pi=\pi^{*}$. Then, 
\begin{align}
\tilde{p}_{\text{error}}(\mathsf{A}_{\text{MAP}}) & \leq\tilde{p}_{\text{error}}(\mathsf{A}^{\sharp})\\
 & =\P[{\cal E}_{A}]\\
 & =\P[{\cal E}_{A}\cap{\cal D}^{c}]+\P[{\cal E}_{A}\cap{\cal D}]\\
 & \leq\P[{\cal E}_{A}\cap{\cal D}^{c}]+\P[{\cal D}]\\
 & =\P[{\cal E}_{A}\cap{\cal D}^{c}\cap{\cal G}_{\eta}]+\P[{\cal E}_{A}\cap{\cal D}^{c}\cap{\cal G}_{\eta}^{c}]+\P[{\cal D}]\\
 & \leq\P[{\cal E}_{A}\cap{\cal D}^{c}\cap{\cal G}_{\eta}]+\P[{\cal G}_{\eta}^{c}]+\P[{\cal D}]\\
 & =\P[{\cal E}_{A}\cap{\cal E}_{C}\cap{\cal D}^{c}\cap{\cal G}_{\eta}]+\P[{\cal E}_{A}\cap{\cal E}_{C}^{c}\cap{\cal D}^{c}\cap{\cal G}_{\eta}]+\P[{\cal G}_{\eta}^{c}]+\P[{\cal D}]\\
 & \leq\P[{\cal E}_{C}]+\P[{\cal E}_{A}\cap{\cal E}_{C}^{c}\cap{\cal D}^{c}\cap{\cal G}_{\eta}]+\P[{\cal G}_{\eta}^{c}]+\P[{\cal D}]\\
 & \leq p_{\text{error}}(\mathsf{C}_{\text{MAP}})+\P[{\cal E}_{A}\cap{\cal E}_{C}^{c}\cap{\cal D}^{c}\cap{\cal G}_{\eta}]+\P[{\cal G}_{\eta}^{c}]+\P[{\cal D}]\\
 & \trre[=,a]p_{\text{error}}(\mathsf{C}_{\text{MAP}})+\P\left[\bigcup_{i=1}^{m}\left\{ \mathsf{D}(\{Y_{j}\}_{j\in{\cal \tilde{J}}_{i}})\neq X_{\pi^{*}(i)}\right\} \cap{\cal E}_{C}^{c}\cap{\cal D}^{c}\cap{\cal G}_{\eta}\right]+\P[{\cal G}_{\eta}^{c}]+\P[{\cal D}]\\
 & =p_{\text{error}}(\mathsf{C}_{\text{MAP}})+\P\left[\bigcup_{i=1}^{m}\left\{ \mathsf{\mathsf{D}}(\{Y_{j}\}_{j\in{\cal {\cal J}}_{i}})\neq X_{i}\right\} \cap{\cal E}_{C}^{c}\cap{\cal D}^{c}\cap{\cal G}_{\eta}\right]+\P[{\cal G}_{\eta}^{c}]+\P[{\cal D}]\\
 & \trre[\leq,b]p_{\text{error}}(\mathsf{C}_{\text{MAP}})+\sum_{i=1}^{m}\P\left[\left\{ \mathsf{\mathsf{D}}(\{Y_{j}\}_{j\in{\cal {\cal J}}_{i}})\neq X_{i}\right\} \cap{\cal E}_{C}^{c}\cap{\cal D}^{c}\cap{\cal G}_{\eta}\right]+\P[{\cal G}_{\eta}^{c}]+\P[{\cal D}]\\
 & \leq p_{\text{error}}(\mathsf{C}_{\text{MAP}})+\sum_{i=1}^{m}\P\left[\left\{ \mathsf{\mathsf{D}}(\{Y_{j}\}_{j\in{\cal {\cal J}}_{i}})\neq X_{i}\right\} \cap{\cal G}_{\eta}\right]+\P[{\cal G}_{\eta}^{c}]+\P[{\cal D}],\label{eq: from assingment error to clustering error}
\end{align}
where $(a)$ follows from the definition of an assignment error, and
$(b)$ follows from the union bound. Now, note that $R_{i}=|{\cal J}_{i}|$
(according to the definition in the proof of Proposition \ref{prop: Raw upper bound on the error probability}),
and let $\overline{X}\sim P_{X}^{\otimes\ell}$ and $\overline{Y}_{1}^{R_{i}}=(\overline{Y}_{1},\overline{Y}_{2},\ldots,\overline{Y}_{r_{i}})$
be drawn conditionally i.i.d. as $\overline{Y}_{j}\sim W^{\otimes\ell}(\cdot\mid\overline{X})$.
Then, the probability of the event in the summation (\ref{eq: from assingment error to clustering error})
is given by 
\[
\P\left[\left\{ \mathsf{D}(\{Y_{j}\}_{j\in{\cal {\cal J}}_{i}})\neq X_{i}\right\} \cap{\cal G}_{\eta}\right]=\P\left[\left\{ \mathsf{\mathsf{D}}(\overline{Y}_{1}^{R_{i}})\neq\overline{X}\right\} \cap{\cal G}_{\eta}\right].
\]
In words, this is the probability that a random source sequence $\overline{X}$
is observed $R_{i}$ times through the reading channel $W^{\otimes\ell}(\cdot\mid\overline{X})$
to obtain $R_{i}$ reads $\overline{Y}_{1}^{R_{i}}$, and the detector
does not detect $\overline{X}$ from these reads. The optimal detector
is given by the MAP rule 
\begin{align}
\hat{\overline{X}}(\overline{Y}_{1}^{R_{i}}) & =\argmax_{x\in{\cal X}^{\otimes\ell}}\P\left[\overline{X}=x\mid\overline{Y}_{1}^{R_{i}}\right]\\
 & =\argmax_{x\in{\cal X}^{\otimes\ell}}P_{X}^{\otimes\ell}(x)\cdot\prod_{j=1}^{R_{i}}W^{\otimes\ell}(\overline{Y}_{j}\mid x).
\end{align}
This is a variant of the multi-view problem by Levenshtein \cite[Sec. IV]{levenshteinEfficientReconstructionSequences2001},
which considered a fixed $\overline{x}$ and imperfect reconstruction.
To upper bound its error probability, let us define for $a,\tilde{a}\in{\cal X}$
the \emph{Chernoff distance} as
\[
d_{C}(a,\tilde{a}):=-\min_{\lambda\in[0,1]}\log\sum_{y\in{\cal Y}}W^{\lambda}(y\mid a)W^{1-\lambda}(y\mid\tilde{a}),
\]
and extend it (subadditively) to $x,\tilde{x}\in{\cal X}^{\otimes\ell}$
for memoryless channels $W^{\otimes\ell}$ as 
\[
D_{C}^{(\ell)}(x,\tilde{x}):=-\min_{\lambda\in[0,1]}\log\sum_{y\in{\cal Y}^{\otimes\ell}}[W^{\otimes\ell}(y\mid x)]^{\lambda}[W^{\otimes\ell}(y\mid\tilde{x})]^{1-\lambda}.
\]
Further let the minimal Chernoff distance between any pair of source
sequences be 
\[
D_{C,\text{min}}^{(\ell)}:=\min_{x,\tilde{x}\in{\cal X}^{\otimes\ell}\colon x\neq\tilde{x}}D_{C}^{(\ell)}(x,\tilde{x}).
\]
For memoryless channels $W^{\otimes\ell}$, it can be  shown that
the minimum is achieved when $d_{\text{Ham}}(x,\tilde{x})=1$ and
then 
\begin{equation}
D_{C,\text{min}}^{(\ell)}=\min_{a,\tilde{a}\in{\cal X}\colon a\neq\tilde{a}}d_{C}(a,\tilde{a})=:d_{C,\text{min}}.\label{eq: minimal Chernoff distance}
\end{equation}
Since $d_{C}(a,\tilde{a})>d_{B}(a,\tilde{a})$, our assumption $d_{B,\text{min}}>0$
implies that $d_{C,\text{min}}>0$. Now, recently, in \cite[Theorem 3.1]{rameshwar2024information},
it was shown that 
\begin{equation}
I(\overline{X};\overline{Y}_{1}^{r_{i}})=H(\overline{X})-\exp\left[-r_{i}d_{C,\text{min}}+c\cdot\log(n|{\cal X}|^{\ell}))\right],\label{eq: mutual information asymptotic for multi-view channel}
\end{equation}
where $d_{C,\text{min}}$ is as defined in (\ref{eq: minimal Chernoff distance}).
As was proved in \cite{baladova1966minimum}, \cite[Eq. (12)]{chu1966inequalities},
\cite[Eq. (41)]{hellman1970probability} (see \cite{sason2017arimoto}
for a survey and refined results), it holds that the error probability
in the detection problem of $\overline{X}$ given $\overline{Y}_{1}^{r_{i}}$
is upper bounded by half the conditional entropy $H(\overline{X}\mid\overline{Y}_{1}^{r_{i}})$.
This and (\ref{eq: mutual information asymptotic for multi-view channel})
thus imply that 
\[
\P\left[\mathsf{\mathsf{D}}(\overline{Y}_{1}^{r_{i}})\neq\overline{X}\right]\leq\frac{1}{2}H(\overline{X}\mid\overline{Y}_{1}^{r_{i}})\leq\frac{1}{2}\exp\left[-r_{i}d_{C,\text{min}}+c\cdot\log(n|{\cal X}|^{\otimes\ell}))\right].
\]
Under the event ${\cal G}_{\eta}$, it further holds that $R_{i}\geq\frac{n}{m}(1-\eta)$
and hence 
\[
\P\left[\left\{ \mathsf{D}(\{Y_{j}\}_{j\in{\cal {\cal J}}_{i}})\neq X_{i}\right\} \cap{\cal G}_{\eta}\right]\leq\frac{1}{2}\exp\left[-\frac{n}{m}(1-\eta)d_{C,\text{min}}+c\cdot\log(n|{\cal X}|^{\otimes\ell}))\right].
\]
Using this bound and the bounds $\P[R_{1}^{m}\not\in{\cal G}_{\eta}]\leq me^{-c\eta^{2}n}$
and $\P[{\cal D}]\leq m^{2}e^{-\ell H_{2}(P_{X})}$ in (\ref{eq: from assingment error to clustering error})
establishes
\[
p_{\text{error}}(\mathsf{C}_{\text{MAP}})\ge\tilde{p}_{\text{error}}(\mathsf{A}_{\text{MAP}})-m^{2}e^{-\ell H_{2}(P_{X})}-e^{-\Theta(n)}-\frac{m}{2}\exp\left[-\frac{n}{m}(1-\eta)d_{C,\text{min}}+c\cdot\log(n|{\cal X}|^{\otimes\ell}))\right],
\]
and when $H_{2}(P_{X})>0$ and $\eta>0$ is arbitrarily small, this
establishes the claim of the proposition. 
\end{IEEEproof}
We may now prove Theorem \ref{thm: upper bound on the exponent}.
\begin{IEEEproof}[Proof of Theorem \ref{thm: upper bound on the exponent}]
Using Proposition \ref{prop: asignment versus clustering error probability},
we lower bound $\tilde{p}_{\text{error}}(\mathsf{A}_{\text{MAP}})$.
Conditioned on $X_{1}^{m}=x_{1}^{m}$, the $n$ reads $Y_{1}^{n}$
are independent, and $\mathsf{A}_{\text{MAP}}$ decides on the assignment
of each read to $i\in[m]$ separately, as 
\begin{equation}
\hat{S}_{j}=\argmax_{i\in[m]}W^{\otimes\ell}(y_{j}\mid x_{i}).\label{eq: ML for read assignment}
\end{equation}
This is a maximum likelihood (ML) rule, since the index $i\in[m]$
is chosen uniformly over $[m]$ (and independently among the $n$
reads). Let ${\cal E}$ denote the total assignment error event of
the rule $\mathsf{A}_{\text{MAP}}$, and let ${\cal E}_{j}$ denote
an assignment error event in the $j$th read, i.e.,  
\[
{\cal E}_{j}:=\left\{ \hat{S}_{j}\neq S_{j}\right\} .
\]
Then, 
\begin{align}
\tilde{p}_{\text{error}}(\mathsf{A}_{\text{MAP}}) & \geq\sum_{x_{1}^{m}\in({\cal X}^{\otimes\ell})^{\otimes m}}\P\left[X_{1}^{m}=x_{1}^{m}\right]\cdot\P\left[{\cal E}\mid X_{1}^{m}=x_{1}^{m}\right]\\
 & =\sum_{x_{1}^{m}\in({\cal X}^{\otimes\ell})^{\otimes m}}\P\left[X_{1}^{m}=x_{1}^{m}\right]\cdot\P\left[\bigcup_{j\in[n]}{\cal E}_{j}\mid X_{1}^{m}=x_{1}^{m}\right]\\
 & \trre[\geq,a]\sum_{x_{1}^{m}\in({\cal X}^{\otimes\ell})^{\otimes m}}\P\left[X_{1}^{m}=x_{1}^{m}\right]\cdot\left[\frac{1}{2}\left(\sum_{j=1}^{n}\P[{\cal E}_{j}\mid X_{1}^{m}=x_{1}^{m}]\right)\wedge1\right]\\
 & =\sum_{x_{1}^{m}\in({\cal X}^{\otimes\ell})^{\otimes m}}\P\left[X_{1}^{m}=x_{1}^{m}\right]\cdot\left[\frac{1}{2}\left(n\cdot\P[{\cal E}_{1}\mid X_{1}^{m}=x_{1}^{m}]\right)\wedge1\right]\\
 & \geq\frac{1}{2}\sum_{x_{1}^{m}\in({\cal X}^{\otimes\ell})^{\otimes m}}\P\left[X_{1}^{m}=x_{1}^{m}\right]\cdot\left[\left(n\cdot\P[{\cal E}_{1}\mid X_{1}^{m}=x_{1}^{m}]\right)\wedge1\right],\label{eq: first lower bound on the assignment MAP error probability}
\end{align}
 where $(a)$ follows since the (clipped) union bound is tight up
to a factor of $2$ for pairwise independent events (e.g., \cite[Appendix A.2]{shulman2003communication}),
that is,
\[
\P[\cup{\cal F}_{j}]\geq\frac{1}{2}\left(\sum\P[{\cal F}_{j}]\right)\wedge1.
\]
Now, consider the set 
\[
{\cal V}_{n}:=\left\{ x_{1}^{m}\in({\cal X}^{\otimes\ell})^{\otimes m}\colon\P[{\cal E}_{1}\mid X_{1}^{m}=x_{1}^{m}]<\frac{1}{n}\right\} ,
\]
so we continue the lower bound in (\ref{eq: first lower bound on the assignment MAP error probability})
as 
\begin{align}
 & \frac{1}{2}\sum_{x_{1}^{m}\in({\cal X}^{\otimes\ell})^{\otimes m}}\P\left[X_{1}^{m}=x_{1}^{m}\right]\cdot\left[\left(n\cdot\P[{\cal E}_{1}\mid X_{1}^{m}=x_{1}^{m}]\right)\wedge1\right]\\
 & =\frac{1}{2}\sum_{x_{1}^{m}\in{\cal V}_{n}}\P\left[X_{1}^{m}=x_{1}^{m}\right]\cdot\left[\left(n\cdot\P[{\cal E}_{1}\mid X_{1}^{m}=x_{1}^{m}]\right)\wedge1\right]+\nonumber \\
 & \hphantom{===}\frac{1}{2}\sum_{x_{1}^{m}\in{\cal V}_{n}^{c}}\P\left[X_{1}^{m}=x_{1}^{m}\right]\cdot\left[\left(n\cdot\P[{\cal E}_{1}\mid X_{1}^{m}=x_{1}^{m}]\right)\wedge1\right]\\
 & \geq\frac{n}{2}\sum_{x_{1}^{m}\in{\cal V}_{n}}\P\left[X_{1}^{m}=x_{1}^{m}\right]\cdot\P[{\cal E}_{1}\mid X_{1}^{m}=x_{1}^{m}]\\
 & =\frac{n}{2}\sum_{x_{1}^{m}\in({\cal X}^{\otimes\ell})^{\otimes m}}\P\left[X_{1}^{m}=x_{1}^{m}\right]\cdot\P[{\cal E}_{1}\mid X_{1}^{m}=x_{1}^{m}]-\frac{n}{2}\sum_{x_{1}^{m}\in{\cal V}_{n}^{c}}\P\left[X_{1}^{m}=x_{1}^{m}\right]\cdot\P[{\cal E}_{1}\mid X_{1}^{m}=x_{1}^{m}]\\
 & =\frac{n}{2}\P[{\cal E}_{1}]-\frac{n}{2}\sum_{x_{1}^{m}\in{\cal V}_{n}^{c}}\P\left[X_{1}^{m}=x_{1}^{m}\right]\cdot\P[{\cal E}_{1}\mid X_{1}^{m}=x_{1}^{m}]\\
 & \geq\frac{n}{2}\P[{\cal E}_{1}]-\frac{n}{2}\P\left[X_{1}^{m}\in{\cal V}_{n}^{c}\right].\label{eq: second lower bound on the assignment MAP error probability}
\end{align}
For the event ${\cal E}_{1}$, since $X_{1}^{m}$ are randomly drawn
from $(P_{X}^{\otimes\ell})^{\otimes m}$, and then an index $S_{1}\in[m]$
is randomly chosen uniformly over $[m]$, and then $Y_{1}\sim W^{\otimes\ell}(\cdot\mid X_{S_{1}})$,
this is exactly the same probabilistic setting as the random coding
analysis for $m$ codewords and the noisy channel $W^{\otimes\ell}$,
at rate $R=\frac{\log m}{\ell}=\frac{\log m}{\beta\log n}=o(1)$.
It is well known that Gallager's random coding bound is tight \cite{gallager1973random}
and so
\begin{equation}
\P[{\cal E}_{1}]\geq\exp\left[-\ell E_{r}(P_{X},0)+o(\ell)\right].\label{eq: Gallagers random coding lower bound for zero rate}
\end{equation}
In order to upper bound $\P[X_{1}^{m}\in{\cal V}_{n}^{c}]$, recall
the definition of the set ${\cal D}_{\tau}$ in (\ref{eq: high probability set for a uniform source})
\begin{equation}
{\cal D}_{\tau}:=\left\{ x_{1}^{m}\in({\cal X}^{\otimes\ell})^{\otimes m}\colon\min_{i\neq j}\frac{1}{\ell}D_{B}(x_{i},x_{j})\geq\tau\right\} 
\end{equation}
Then, letting $\delta>0$ be arbitrarily small and choosing $\tau=\frac{1}{\beta}+\delta$,
it holds that for any $x_{1}^{m}\in{\cal D}_{\frac{1}{\beta}+\delta}$
\begin{align}
 & \P\left[{\cal E}_{1}\mid X_{1}^{m}=x_{1}^{m}\right]\nonumber \\
 & \trre[\leq,a]\P\left[\bigcup_{\tilde{i}\in[m]\colon\tilde{i}\neq i}\left\{ W^{\otimes\ell}(y_{j}\mid x_{\tilde{i}})\geq W^{\otimes\ell}(y_{j}\mid x_{i})\right\} \right]\\
 & \trre[\leq,b]\sum_{\tilde{i}\in[m]\colon\tilde{i}\neq i}\P\left[W^{\otimes\ell}(y_{j}\mid x_{\tilde{i}})\geq W^{\otimes\ell}(y_{j}\mid x_{i})\right]\\
 & \trre[\leq,c]\sum_{\tilde{i}\in[m]\colon\tilde{i}\neq i}\exp\left[-D_{B}(x_{i},x_{\tilde{i}})\right]\\
 & \trre[\leq,d]\sum_{\tilde{i}\in[m]\colon\tilde{i}\neq i}\exp\left[-\ell\left(\frac{1}{\beta}+\delta\right)\right]\\
 & \trre[\leq,e]m\exp\left[\log\frac{1}{n}-\delta\ell\right]\\
 & =\frac{1}{n}\cdot e^{-\delta\ell+o(\ell)},
\end{align}
where $(a)$ follows from the ML decision rule (\ref{eq: ML for read assignment}),
$(b)$ follows from the union bound, and $(c)$ follows from the standard
Bhattacharyya upper bound, $(d)$ follows from the assumption $x_{1}^{m}\in{\cal D}_{\frac{1}{\beta}+\delta}$,
and $(e)$ since $\ell=\beta\log n$. Hence, for all $\ell$ (or $n$)
large enough, which only depends on $(m,\delta)$ but not on $x_{1}^{m}$,
it holds that if $x_{1}^{m}\in{\cal D}_{\frac{1}{\beta}+\delta}$
then $x_{1}^{m}\in{\cal V}_{n}$. So, for all such large enough $n$,
\begin{align}
\P\left[X_{1}^{m}\in{\cal V}_{n}^{c}\right] & \leq\P\left[X_{1}^{m}\in{\cal D}_{\frac{1}{\beta}+\delta}^{c}\right]\\
 & \leq\exp\left[-E_{\text{B}}\left(\frac{1}{\beta}+\delta\right)\cdot\ell+o(\ell)\right],
\end{align}
as we have seen in (\ref{eq: large deviations for set F uniform}).
Substituting this bound and (\ref{eq: Gallagers random coding lower bound for zero rate})
into (\ref{eq: second lower bound on the assignment MAP error probability}),
and then into (\ref{eq: first lower bound on the assignment MAP error probability}),
shows that 
\[
\tilde{p}_{\text{error}}(\mathsf{A}_{\text{MAP}})\geq\frac{n}{2}\exp\left[-\ell E_{r}(P_{X},0)+o(\ell)\right]-\frac{n}{2}\exp\left[-\ell E_{\text{B}}\left(\frac{1}{\beta}+\delta\right)+o(\ell)\right].
\]
As this holds for any $\delta>0$ we may utilize the continuity of
$\tau\to E_{\text{B}}(\tau)$ (see (\ref{eq: exponent of Bhatacharrya}))
and $\ell=\beta\log n$ to obtain
\[
\tilde{p}_{\text{error}}(\mathsf{A}_{\text{MAP}})\geq\exp\left[-\ell\left(E_{r}(P_{X},0)-\frac{1}{\beta}\right)+o(\ell)\right]-\exp\left[-\ell\left(E_{\text{B}}\left(\frac{1}{\beta}\right)-\frac{1}{\beta}\right)+o(\ell)\right].
\]
Now, we assumed that $d_{B,\text{min}}>0$ and since for any $a,\tilde{a}\in{\cal X}$
it holds that $d_{C}(a,\tilde{a})>d_{B}(a,\tilde{a})$ it also holds
that $d_{C,\text{min}}>0$. Hence, the assumption of Proposition \ref{prop: asignment versus clustering error probability}
holds and it then implies that 
\begin{align}
p_{\text{error}}(\mathsf{C}_{\text{MAP}}) & \ge\tilde{p}_{\text{error}}(\mathsf{A}_{\text{MAP}})-e^{-\ell H_{2}(P_{X})+o(\ell)}\\
 & \geq e^{-\ell\left(E_{r}(P_{X},0)-\frac{1}{\beta}\right)+o(\ell)}-e^{-\ell\left(E_{\text{B}}\left(\frac{1}{\beta}\right)-\frac{1}{\beta}\right)+o(\ell)}-e^{-\ell H_{2}(P_{X})+o(\ell)},
\end{align}
from which the statement of the theorem follows immediately by using
$E_{r}(P_{X},0)=-\log B(P_{X}\times P_{X})$ and noting that the last
two terms are exponentially smaller than the first one, given the
premise of the theorem. 
\end{IEEEproof}

\section{Useful Lemmas \label{sec:Useful-Lemmas}}
\begin{lem}
\label{lem: Input distribution continuity}Assume that $p_{\text{\emph{min}}}=\min_{x\in{\cal X}}P_{X}(x)>0$,
and let $\delta>0$ be given. Then, 
\[
\max_{\tilde{x}\in{\cal X}^{\otimes\ell}\colon d_{\text{\emph{Ham}}}(x,\tilde{x})\leq\delta\ell}P_{X}^{\otimes\ell}(\tilde{x})\leq e^{-\delta\ell\log p_{\text{\emph{min}}}}\cdot P_{X}^{\otimes\ell}(x)
\]
\end{lem}
\begin{IEEEproof}
The claim holds since 
\begin{align}
\log P_{X}^{\otimes\ell}(\tilde{x}) & =\sum_{i=1}^{\ell}\log P_{X}(\tilde{x}(i))\\
 & =\sum_{i=1}^{\ell}\log P_{X}(x(i))+\log\frac{P_{X}(\tilde{x}(i))}{P_{X}(x(i))}\\
 & =\log P_{X}^{\otimes\ell}(x)+\sum_{i\in[\ell]\colon\tilde{x}(i)\neq x(i)}\log\frac{P_{X}(\tilde{x}(i))}{P_{X}(x(i))}\\
 & \leq\log P_{X}^{\otimes\ell}(x)+\delta\ell\cdot\log\frac{1}{p_{\text{min}}}
\end{align}
\end{IEEEproof}
\begin{lem}
\label{lem: Bhat Hamming ball property}Assume that $d_{B,\text{\emph{max}}}<\infty$.
Let $\delta>0$ be given. For any $x_{1},x_{2}\in{\cal X}^{\otimes\ell}$
\[
\min_{\tilde{x}_{1}\in{\cal X}^{\otimes\ell}\colon d_{\text{\emph{Ham}}}(x_{1},\tilde{x}_{1})\leq\delta\ell}D_{B}(\tilde{x}_{1},x_{2})\geq D_{B}(x_{1},x_{2})-\delta d_{B,\text{\emph{max}}}\ell.
\]
\end{lem}
\begin{IEEEproof}
The claim holds since 
\begin{align}
D_{B}(\tilde{x}_{1},x_{2}) & =\sum_{i=1}^{\ell}d_{B}(\tilde{x}_{1}(i),x_{2}(i))\\
 & \geq\sum_{i\in[\ell]\colon\tilde{x}_{1}(i)=x_{1}(i)}d_{B}(x_{1}(i),x_{2}(i))\\
 & =\sum_{i=1}^{\ell}d_{B}(x_{1}(i),x_{2}(i))-\delta\ell\cdot d_{B,\text{max}}\\
 & \geq D_{B}(x_{1},x_{2})-\delta\ell\cdot d_{B,\text{max}}.
\end{align}
\end{IEEEproof}
\begin{lem}
\label{lem: large Hamming implies large Bhat }Assume that $d_{B,\text{\emph{min}}}>0$.
Then, 
\[
\min_{x,\tilde{x}\in{\cal X}^{\otimes\ell}\colon d_{\text{\emph{Ham}}}(x,\tilde{x})>\delta\ell}D_{B}(\tilde{x},x)>\delta d_{B,\text{\emph{min}}}\ell.
\]
\end{lem}
\begin{IEEEproof}
Follows from 
\begin{align}
D_{B}(\tilde{x},x) & =\sum_{i=1}^{\ell}d_{B}(\tilde{x}(i),x(i))\\
 & \geq\sum_{i\in[\ell]\colon\tilde{x}(i)\neq x(i)}d_{B}(\tilde{x}(i),x(i))\\
 & >\delta\ell\cdot d_{B,\text{\ensuremath{\min}}}.
\end{align}
\end{IEEEproof}
\begin{lem}
\label{lem: binomial asymptotics}If $\alpha>1$ then as $n\to\infty$
\[
\left(1+\frac{1}{n^{\alpha}}\right)^{n}-1=\frac{1}{n^{\alpha-1}}+O\left(\frac{1}{n^{2\alpha-2}}\right)
\]
\end{lem}
\begin{IEEEproof}
It holds that
\begin{align}
 & \left(1+\frac{1}{n^{\alpha}}\right)^{n}-1\nonumber \\
 & =\exp\left[\log\left[\left(1+\frac{1}{n^{\alpha}}\right)^{n}\right]\right]-1\\
 & =\exp\left[n\log\left(1+\frac{1}{n^{\alpha}}\right)\right]-1\\
 & \trre[=,*]\exp\left[n\left(\frac{1}{n^{\alpha}}-\frac{1}{2n^{2\alpha}}+O\left(\frac{1}{n^{3\alpha}}\right)\right)\right]-1\\
 & \trre[=,**]n\left(\frac{1}{n^{\alpha}}-\frac{1}{2n^{2\alpha}}+O\left(\frac{1}{n^{3\alpha}}\right)\right)+\frac{n^{2}}{2}\left(\frac{1}{n^{2\alpha}}-\frac{1}{2n^{4\alpha}}+O\left(\frac{1}{n^{6\alpha}}\right)\right)+O\left(\frac{n^{3}}{n^{3\alpha}}\right)\\
 & =\frac{1}{n^{\alpha-1}}+O\left(\frac{1}{n^{2\alpha-2}}\right),
\end{align}
where $(*)$ follows from 
\[
\log(1+t)=t-\frac{1}{2}t^{2}+O(t^{3})
\]
and $(**)$ follows from 
\[
e^{t}-1=t+\frac{t^{2}}{2}+O(t^{3}).
\]
\end{IEEEproof}
\begin{lem}
\label{lem:Computation of E_b}Let 
\[
Z(\lambda):=\sum_{x_{1}^{2}\in{\cal X}^{\otimes2}}\frac{P_{X}(x_{1})P_{X}(x_{2})}{B^{\lambda}(x_{1},x_{2})}
\]
It holds that 
\[
E_{\text{B}}(\tau)=-\log Z(\lambda_{\tau}),
\]
where $\lambda_{\tau}\geq0$ is the solution to the equation 
\[
\frac{1}{Z(\lambda)}\sum_{x_{1}^{2}\in{\cal X}^{\otimes2}}\frac{P_{X}(x_{1})P_{X}(x_{2})}{B^{\lambda}(x_{1},x_{2})}d_{B}(x_{1},x_{2})=\tau.
\]
\end{lem}
\begin{IEEEproof}
We write the constrained minimization problem of $E_{\text{B}}(\tau)$
over $Q_{X_{1}^{2}}$ in a Lagrangian form, with multipliers $\lambda\geq0$
for the constraint and $\mu\in\mathbb{R}$ for normalization (ignoring
the positivity constraints as they will be satisfied anyway).
\begin{align}
 & \min_{Q_{X_{1}X_{2}}}\sum_{x_{1}^{2}\in{\cal X}^{\otimes2}}Q_{X_{1}X_{2}}(x_{1},x_{2})\log\frac{Q_{X_{1}X_{2}}(x_{1},x_{2})}{P_{X}(x_{1})P_{X}(x_{2})}\nonumber \\
 & \hphantom{==}+\lambda\left(\sum_{x_{1}^{2}\in{\cal X}^{\otimes2}}Q_{X_{1}X_{2}}(x_{1},x_{2})d_{B}(x_{1},x_{2})\right)+\mu\sum_{x_{1}^{2}\in{\cal X}^{\otimes2}}Q_{X_{1}X_{2}}(x_{1},x_{2})\nonumber \\
 & =\min_{Q_{X_{1}X_{2}}}\sum_{x_{1}^{2}\in{\cal X}^{\otimes2}}Q_{X_{1}X_{2}}(x_{1},x_{2})\log\frac{Q_{X_{1}X_{2}}(x_{1},x_{2})B^{\lambda}(x_{1},x_{2})}{P_{X}(x_{1})P_{X}(x_{2})}+\mu\sum_{x_{1}^{2}\in{\cal X}^{\otimes2}}Q_{X_{1}X_{2}}(x_{1},x_{2}).
\end{align}
Taking derivative w.r.t. $Q_{X_{1}^{2}}(x_{1}^{2})$ for specific
$x_{1}^{2}\in P_{X}^{\otimes2}$ and equating to zero to obtain a
stationary point, we get
\[
Q_{X_{1}^{2}}^{*}(x_{1}^{2})=\frac{P_{X}(x_{1})P_{X}(x_{2})}{B^{\lambda}(x_{1},x_{2})Z(\lambda)}
\]
where $\lambda$ is chosen to satisfy the constraint (and thus depends
on $\tau$), and 
\[
Z(\lambda)=\sum_{x_{1}^{2}\in{\cal X}^{\otimes2}}\frac{P_{X}(x_{1})P_{X}(x_{2})}{B^{\lambda}(x_{1},x_{2})}
\]
 is the normalization constant (partition function). The KL divergence
is then given by $\log\frac{1}{Z(\lambda)}$ where $\lambda\geq0$
is chosen so that 
\[
\frac{1}{Z(\lambda)}\sum_{x_{1}^{2}\in{\cal X}^{\otimes2}}\frac{P_{X}(x_{1})P_{X}(x_{2})}{B^{\lambda}(x_{1},x_{2})}d_{B}(x_{1},x_{2})=\tau.
\]
\end{IEEEproof}
\begin{lem}
\label{lem: Computation of Ep}Let 
\[
Z(\lambda):=\left(\sum_{x\in{\cal X}}P_{X}^{1+(m-1)\lambda}(x)\right)\times\left(\sum_{x\in{\cal X}}P_{X}^{1-\lambda}(x)\right)^{m-1}
\]
It holds that 
\[
E_{\text{P}}(\rho)=-\log Z(\lambda_{\rho}),
\]
where $\lambda_{\rho}\geq0$ is the solution to the equation 
\[
\frac{1}{Z(\lambda)}\sum_{x\in{\cal X}}\left(P_{X}^{1+(m-1)\lambda}(x)-P_{X}^{1-\lambda}(x)\right)\cdot\log P_{X}(x)=\frac{\rho}{m-1}.
\]
\end{lem}
\begin{IEEEproof}
We write the constrained minimization problem of $E_{\text{P}}(\rho)$
over $Q_{X_{1}^{m}}$ in a Lagrangian form, with multipliers $\lambda\geq0$
for the constraint and $\mu\in\mathbb{R}$ for normalization (ignoring
the positivity constraints as they will be satisfied anyway). This
form is as follows
\begin{align}
 & \min_{Q_{X_{1}^{m}}}\sum_{x_{1}^{m}\in{\cal X}^{\otimes m}}Q_{X_{1}^{m}}(x_{1}^{m})\log\frac{Q_{X_{1}^{m}}(x_{1}^{m})}{P_{X}(x_{1})\cdots P_{X}(x_{m})}\nonumber \\
 & \hphantom{==}-\lambda\left(\sum_{x_{1}^{m}\in{\cal X}^{\otimes m}}Q_{X_{1}^{m}}(x_{1}^{m})\log\frac{P_{X}^{m-1}(x_{1})}{P_{X}(x_{2})\cdots P_{X}(x_{m})}\right)+\mu\sum_{x_{1}^{m}\in{\cal X}^{\otimes m}}Q_{X_{1}^{m}}(x_{1}^{m})\nonumber \\
 & =\min_{Q_{X_{1}^{m}}}\sum_{x_{1}^{m}\in{\cal X}^{\otimes m}}Q_{X_{1}^{m}}(x_{1}^{m})\log\frac{Q_{X_{1}^{m}}(x_{1}^{m})}{P_{X}^{1+(m-1)\lambda}(x_{1})P_{X}^{1-\lambda}(x_{2})\cdots P_{X}^{1-\lambda}(x_{m})}+\mu\sum_{x_{1}^{m}\in{\cal X}^{\otimes m}}Q_{X_{1}^{m}}(x_{1}^{m})
\end{align}
Taking derivative w.r.t. $Q_{X_{1}^{m}}(x_{1}^{m})$ for specific
$x_{1}^{m}\in P_{X}^{\otimes m}$ and equating to zero to obtain a
stationary point, we get
\[
Q_{X_{1}^{m}}^{*}(x_{1}^{m})=\frac{P_{X}^{1+(m-1)\lambda}(x_{1})P_{X}^{1-\lambda}(x_{2})\cdots P_{X}^{1-\lambda}(x_{m})}{Z(\lambda)}
\]
where $\lambda$ is chosen to satisfy the constraint (and thus depends
on $\rho$), and $Z(\lambda)$ is the normalization constant (partition
function). Note that 
\begin{align}
Z(\lambda) & =\sum_{x_{1}^{m}\in{\cal X}^{\otimes m}}P_{X}^{1+(m-1)\lambda}(x_{1})P_{X}^{1-\lambda}(x_{2})\cdots P_{X}^{1-\lambda}(x_{m})\\
 & =\sum_{x_{1}\in{\cal X}}P_{X}^{1+(m-1)\lambda}(x_{1})\sum_{x_{2}\in{\cal X}}P_{X}^{1-\lambda}(x_{2})\cdots\sum_{x_{m}\in{\cal X}}P_{X}^{1-\lambda}(x_{m})\\
 & =\left(\sum_{x\in{\cal X}}P_{X}^{1+(m-1)\lambda}(x)\right)\times\left(\sum_{x\in{\cal X}}P_{X}^{1-\lambda}(x)\right)^{m-1}
\end{align}
The KL divergence is then given by 
\[
\log\frac{1}{Z(\lambda)}=-\log\left(\sum_{x\in{\cal X}}P_{X}^{1+(m-1)\lambda}(x)\right)-(m-1)\log\left(\sum_{x\in{\cal X}}P_{X}^{1-\lambda}(x)\right)
\]
where $\lambda\geq0$ is chosen so that 
\[
\sum_{x_{1}^{m}\in{\cal X}^{\otimes m}}Q_{X_{1}^{m}}^{*}(x_{1}^{m})\log\frac{P_{X}^{m-1}(x_{1})}{P_{X}(x_{2})\cdots P_{X}(x_{m})}=\rho,
\]
that is 
\[
\sum_{x_{1}\in{\cal X}}Q_{X_{1}}^{*}(x_{1})\log P_{X}^{m-1}(x_{1})-\sum_{i=2}^{m}\sum_{x_{i}\in{\cal X}}Q_{X_{i}}^{*}(x_{i})\log P_{X}(x_{i})=\rho,
\]
or, from symmetry of $Q_{X_{i}}^{*}(x_{i})$ for $i\in[m]\backslash\{1\}$,
\[
\sum_{x\in{\cal X}}Q_{X_{1}}^{*}(x)\log P_{X}(x)-\sum_{x\in{\cal X}}Q_{X_{2}}^{*}(x)\log P_{X}(x)=\frac{\rho}{m-1},
\]
that is, 
\[
\frac{1}{Z(\lambda)}\sum_{x\in{\cal X}}\left(P_{X}^{1+(m-1)\lambda}(x)-P_{X}^{1-\lambda}(x)\right)\cdot\log P_{X}(x)=\frac{\rho}{m-1},
\]
as claimed. 
\end{IEEEproof}
\bibliographystyle{plain}
\bibliography{DNA_clustering}

\end{document}